\documentclass[journal]{IEEEtran}

\newcommand{\revise}[1]{{\color{blue}#1}}

\ifCLASSINFOpdf
\else
\fi
\ifCLASSOPTIONcompsoc
 \usepackage[caption=false,font=normalsize,labelfont=sf,textfont=sf]{subfig}
\else
 \usepackage[caption=false,font=footnotesize]{subfig}
\fi
\begin{document}
%
\title{Over-the-Air Federated Multi-Task Learning Over MIMO Multiple Access Channels}
%
%
%

\author{Chenxi~Zhong,
        Huiyuan~Yang,
        and~Xiaojun~Yuan,~\IEEEmembership{Senior~Member,~IEEE}
    \thanks{C. Zhong, H. Yang and X. Yuan are with the National Key Laboratory of Science and Technology on Communications, University of Electronic Science and Technology of China, Chengdu, China (e-mail: cxzhong@std.uestc.edu.cn; hyyang@std.uestc.edu.cn; xjyuan@uestc.edu.cn). The corresponding author is Xiaojun Yuan.}
}

\maketitle

\begin{abstract}
With the explosive growth of data and wireless devices, federated learning (FL) over wireless medium has emerged as a promising technology for large-scale distributed intelligent systems. Yet, the urgent demand for ubiquitous intelligence will generate a large number of concurrent FL tasks, which may seriously aggravate the scarcity of communication resources. By exploiting the analog superposition of electromagnetic waves, over-the-air computation (AirComp) is an appealing solution to alleviate the burden of communication required by FL. However, sharing frequency-time resources in over-the-air computation inevitably brings about the problem of inter-task interference, which poses a new challenge that needs to be appropriately addressed. In this paper, we study over-the-air federated multi-task learning (OA-FMTL) over the multiple-input multiple-output (MIMO) multiple access (MAC) channel. We propose a novel model aggregation method for the alignment of local gradients of different devices, which alleviates the straggler problem in over-the-air computation due to the channel heterogeneity. We establish a communication-learning analysis framework for the proposed OA-FMTL scheme by considering the spatial correlation between devices, and formulate an optimization problem for the design of transceiver beamforming and device selection. To solve this problem, we develop an algorithm by using alternating optimization (AO) and fractional programming (FP), which effectively mitigates the impact of inter-task interference on the FL learning performance. We show that due to the use of the new model aggregation method, device selection is no longer essential, thereby avoiding the heavy computational burden involved in selecting active devices. Numerical results demonstrate the validity of the analysis and the superb performance of the proposed scheme. 
\end{abstract}


\section{Introduction}
In recent years, the explosive increase of wireless data has promoted widespread applications of artificial intelligence in many fields, such as computer vision \cite{he2016deep} and natural language processing \cite{young2018recent}. To exploit the diversity of wireless data, the conventional centralized machine learning (ML) paradigm requires edge devices to upload their local data to a central parameter server (PS) for joint training. However, data uploading brings about huge costs in communication resources and potentially threatens user privacy. To tackle these challenges, federated learning (FL) has been proposed as a promising substitute \cite{konevcny2016federated}. In FL, the PS first broadcasts the global model parameters to the selected edge devices. Subsequently, each selected device calculates a local gradient based on its local dataset and upload the gradient to the PS. The global model is then updated by the PS based on the received local gradients. Clearly, by replacing data sharing with gradient sharing, FL reduces communication costs and protects user privacy.

Despite the appealing advantages, the huge uplink communication cost is still a critical bottleneck for FL, especially when uploading through wireless channels. Recently, over-the-air computation (AirComp) has been adopted to improve the communication efficiency of the FL uplink. \cite{nazer2007computation} In over-the-air FL (OA-FL), edge devices transmit their local gradients by using shared radio resources, and aggregate them over the air by utilizing the superposition property of electromagnetic waves. Pioneering work has confirmed that over-the-air FL has strong noise tolerance \cite{zhu2020one} and reduces latency substantially compared with the schemes based on conventional orthogonal multiple access (OMA) protocols \cite{zhu2020broadband, liu2020reconfigurable,zhu2018mimo,liu2021CSIT-Free}. 

However, given the above mentioned benefits, introducing AirComp into FL uplink also brings some tricky problems, such as the so-called straggler\footnote{\revise{ In computer science literature, the word ``straggler'' usually refers to a device with low computation capacity. But here a straggler refers to a device with poor channel condition.} } problem \cite{zhu2020broadband,liu2020reconfigurable}. To aggregate the local gradients at the PS, the devices with better channel conditions have to lower their transmitting powers to align themselves with the stragglers (the devices with the worst channel conditions); see e.g. the state-of-the-art OA-FL literature \cite{liu2020reconfigurable, yang2020federated, shi2021over}. Clearly, since this strict-alignment approach potentially reduces the aggregation signal-to-noise ratio (SNR), it conceivably leads to a degradation of the FL performance. Refs.~\cite{zhu2020broadband} and \cite{liu2020reconfigurable} propose to exclude the stragglers from model aggregation to alleviate this problem. However, this exclusion of devices reduces the available dataset for FL training, thereby deteriorating the FL performance. Thus, developing more efficient approaches to deal with the straggler problem is highly desirable.


Another problem brought by the over-the-air FL uplink appears in the multi-tasking situation. Specifically, 
when multiple tasks are performed simultaneously at the wireless edge\footnote{
The fast development of intelligent systems spawns a large number of FL tasks to meet various demands of intelligence \cite{liu2020federated}.
}, the communication cost of this multi-task FL system will be further exacerbated, which to a greater extent necessitates the use of AirComp. However, introducing AirComp in this multi-task situation implies that the FL tasks share time-frequency resources for gradients uploading, which inevitably introduces inter-task interference. In this paper, we attempt to suppress the inter-task interference by using the multi-antenna technique. 


Moreover, the local gradients of FL exhibit strong correlations temporally (over communication rounds) and spatially (among devices) \cite{han2015learning, chen2021scalecom}.
The temporal correlation has been exploited to improve the uplink efficiency of OA-FL in \cite{fan2021temporal}, whereas an efficient use of the spatial correlation has not been well explored. The existing approaches in OA-FL \cite{cao2021cooperative, yang2020federated} ignore this correlation and assume spatially independent local gradients in the system optimization, potentially leading to a substantial performance loss. Developing a spatial-correlation-aware OA-FL system optimization scheme will be a non-trivial step to exploit the spatial correlation of the local gradients.



\revise{ In this paper, we address the above three challenges by developing a novel over-the-air federated multi-task learning (OA-FMTL) scheme where multiple FL tasks are trained simultaneously over a multiple-input multiple-output (MIMO) multiple access (MAC) channel.} 
\revise{ The MIMO MAC channel consists of a multi-antenna central PS and a number of multi-antenna edge devices.}
\revise{ These FL tasks share time-frequency resources and thus generally interfere with each other in model uploading.} 
\revise{ We establish a communication-learning analysis framework for the considered OA-FMTL scheme. Based on the framework, we analyse the convergence of the OA-FMTL scheme and formulate a transceiver beamforming problem for learning performance enhancement and inter-task interference suppression. We propose a low-complexity algorithm based on alternating optimization (AO) and fractional programming (FP) to optimize the transmit and receive beamforming vectors. The main novelties of our approach are listed as follows:
    \begin{itemize}
        \item \textit{The OA-FMTL analysis framework}: We establish a communication-learning analysis framework for the considered OA-FMTL scheme. Based on the analysis results, we formulate the transceiver beamforming problem and develop a low-complexity solution to the problem.
        \item \textit{Inter-task interference suppression}: To the best of our knowledge, this is the first attempt to solve the problem of inter-task interference suppression in OA-FMTL utilizing the MIMO technique. 
        \item \textit{Spatial-correlation-aware design}: We establish a probability model to capture the gradient correlation among the devices, allowing us to design the OA-FMTL scheme with the awareness of the spatial correlation. We show that the awareness of the spatial correlation brings substantial improvement in learning performance.
        \item \textit{Soft straggler alignment}: We propose a misalignment-tolerant aggregation approach for the gradient aggregation at the PS. This approach corrects the stereotype that gradients need to be strictly aligned without going to the other extreme, i.e., to drop out the stragglers, thereby significantly relieving the performance degradation due to the straggler problem.
    \end{itemize}
}
Extensive numerical results demonstrate that our proposed OA-FMTL scheme achieves significant performance improvements than the existing schemes. Furthermore, the test accuracies of the proposed scheme are very close to those of the ideal error-free benchmarks even in the presence of severe inter-task interference, demonstrating the effectiveness of our proposed scheme.

The remainder of this paper is organized as follows. In Section II, the FL model, the MIMO MAC channel model, the communication system, and the over-the-air model aggregation method are described. In Section III, we analyse the learning performance of the proposed OA-FMTL scheme. In Section IV, we formulate the optimization problem that minimizes the total training loss of the FL tasks and propose algorithms to jointly optimize transmit and receive beamforming vectors, as well as device selection. In Section V, numerical results are presented to evaluate the proposed scheme. Finally, we draw the conclusions in Section VI.

\textit{Notation}: We use $\mathbb{R}$ and $\mathbb{C}$ to denote the real and complex number sets, respectively. We denote scalars in italic type, vectors in straight bold small letters and matrices in straight bold capital letters. $(\cdot)^\dagger$, $(\cdot)^\mathrm{T}$, and $(\cdot)^\mathrm{H}$ are used to denote the conjugate, the transpose, and the conjugate transpose, respectively. We use $s[d]$ to denote the $d$-th entry of vector $\mathbf{s}$, $\mathbf{s}(1:D)$ to denote a sub-vector of $\mathbf{s}$ with entries indexed from $1$ to $D$, $\mathcal{CN}(\mu, \sigma^2)$ to denote the circularly-symmetric complex normal distribution with mean $\mu$ and covariance $\sigma^2$, $\mathbb{E}[\cdot]$ to denote the expectation operator, and $|\mathcal{S}|$ to denote the cardinality of set $\mathcal{S}$. $\mathbf{I}_N$, $\mathbf{1}_{N\times M}$, and $\mathbf{0}_{N\times M}$ are used to denote the $N\times N$ identity matrix, the $N\times M$ all-one matrix, and the $N\times M$ all-zero matrix, respectively. We use $\|\cdot\|$ to denote the $l_2$-norm. $[K]$ denotes the set $\{k | 1 \leq k \leq K\}$.

\section{System Model}
\revise{ In this paper, we consider a $K$-task OA-FMTL system with one central PS and $M$ wireless devices, where $M_k$ wireless devices are collaboratively training an identical model for task $k$\footnote{\revise{ For simplicity, we assume that each device is assigned to a single task. The extension to the case of one device serving multiple tasks is straightforward. We omit detailed discussions due to space limitation.}}, $k\in[K]$, and $M = \sum_{k=1}^K M_k$. A two-task OA-FMTL system is illustrated in Fig.~\ref{fig:system_model}. 
In the following, we introduce the FL system, the underlying communication system, the over-the-air model aggregation, and OA-FMTL framework, sequentially.}

\begin{figure}[htbp]
\centering
\includegraphics[width=0.6\textwidth]{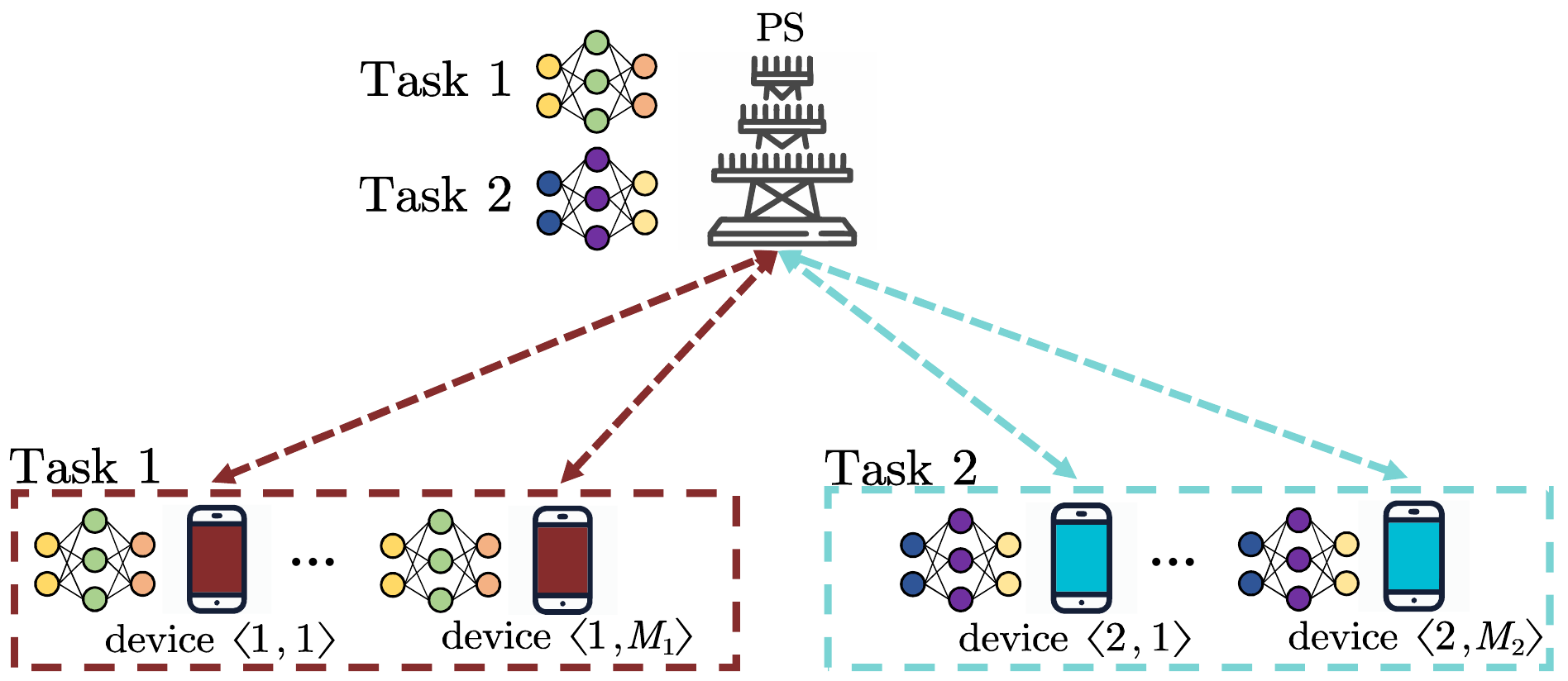}
\caption{FL interference network with two tasks.}
\label{fig:system_model}
\end{figure}

\subsection{Federated Learning System}


We start with the description of the $K$-task FL system. \revise{ Let $\mathcal{A}_k$ denote the training dataset of the $k$-th FL task (which is stored on the $M_k$ devices). Let $Q_k$ be the cardinality of $\mathcal{A}_k$.} Each FL task $k$ has an empirical loss function based on $\mathcal{A}_k$, given by
\begin{equation}
    \label{LossFuncVector}
    F_k(\mathbf{w}_k) \triangleq \frac{1}{Q_k} \sum_{n=1}^{Q_k} f_k\left( \mathbf{w}_k; \boldsymbol{\xi}_{k,n}\right), ~\forall k \in [K],
\end{equation}
where $\mathbf{w}_k \in \mathbb{R}^D$ denotes the parameter vector of task $k$, \footnote{ If the lengths of the parameter vectors vary among the tasks, zero padding is employed to ensure that $\{\mathbf{w}_k\}_{k=1}^K$ have identical length $D$.}, $\boldsymbol{\xi}_{k,n} \in \mathcal{A}_k$ denotes the $n$-th training sample in $\mathcal{A}_k$, and $f_k\left(\mathbf{w}_k; \boldsymbol{\xi}_{k,n}\right)$ denotes the sample-wise loss function with respect to (w.r.t.) $\boldsymbol{\xi}_{k,n}$. Let subscript $\langle k, i \rangle$ denote the index of the $i$-th device in the $k$-th group, $\forall k \in [K]$, $i \in [M_k]$. \revise{We assume that the local training datasets of each device $\langle k,i \rangle$ in group $k$ are independently \revise{and} randomly \revise{drawn} from $\mathcal{A}_k$, $k\in[K]$.} Let $Q_{\langle k, i \rangle}$ denote the sample size of the $\langle k, i \rangle$-th device. Naturally, we have $Q_k = \sum_{i=1}^{M_k} Q_{\langle k, i \rangle}$. Then $F_k(\mathbf{w}_k)$ can be rewritten as
\begin{equation}
    \label{LossFuncTaskk}
    F_k(\mathbf{w}_k)=\frac{1}{Q_k} \sum_{i=1}^{M_k} Q_{\langle k, i \rangle} F_{\langle k, i \rangle}(\mathbf{w}_k), ~\forall k \in [K],
\end{equation}
where $F_{\langle k, i \rangle}(\mathbf{w}_k) \triangleq \sum_{n=1}^{ Q_{\langle k, i \rangle}} f_k \left(\mathbf{w}_k; \boldsymbol{\xi}_{\langle k, i \rangle,n} \right)/Q_{\langle k, i \rangle}$ denotes the local loss function of the $\langle k,i\rangle$-th device with $\boldsymbol{\xi}_{\langle k, i \rangle,n}$ being the $n$-th training sample on the $\langle k,i \rangle$-th device. 
The overall loss function of the considered OA-FMTL is given by
\begin{equation}
    \label{LossFuncWeightedSum}
    \mathcal{F}(\mathbf{w}) \triangleq \sum_{k=1}^K F_k(\mathbf{w}_k),
\end{equation}
where $\mathbf{w} \triangleq [\mathbf{w}_1^\mathrm{T}, \dots, \mathbf{w}_K^\mathrm{T}]^\mathrm{T}$. 

\revise{OA-FMTL aims to minimize $\mathcal{F}(\mathbf{w})$ by separately minimize each $F_k(\mathbf{w}_k)$ using gradient descent (GD) \cite{konevcny2016federated}. 
We now focus on the training of the $k$-th task.}
To be specific, at the $t$-th communication round, the training of each FL task $k$ consists of the following five steps:
\begin{itemize}
    \item \textit{Device selection}: \revise{The parameter server (PS) selects a \revise{sub-group $\mathcal{M}_k$ of group $k$} to participate in the updating.}
	\item \textit{Global model broadcasting}:
    The PS broadcasts \textit{global model} $\mathbf{w}_k^{(t)}$ to the selected devices through error-free links.
    \item \textit{Local gradients computing}: Each selected device computes its gradient based on the local training dataset, \revise{i.e., to compute the local gradient $\mathbf{g}_{\langle k, i \rangle}^{(t)}$ by}
    \begin{equation}
        \label{GradDef}
        \mathbf{g}_{\langle k, i \rangle}^{(t)} \triangleq \nabla F_{\langle k, i \rangle}(\mathbf{w}_k^{(t)}) = \frac{1}{Q_{\langle k, i \rangle}} \sum\nolimits_{n=1}^{ Q_{\langle k, i \rangle}} \nabla f_k\left(\mathbf{w}_k^{(t)}; \boldsymbol{\xi}_{\langle k, i \rangle,n} \right), \revise{\forall i \in \mathcal{M}_k.}
    \end{equation}
    \item \textit{Local gradients uploading}: Selected devices upload their local gradients to the PS.
    \item \textit{Global model updating}: \revise{The PS computes the aggregated gradient $\mathbf{g}_k^{(t)}$ by
    \begin{equation}
        \label{idealgradient}
        \mathbf{g}_k^{(t)} \triangleq \frac{ \sum\nolimits_{i \in \mathcal{M}_k} Q_{\langle k, i \rangle} \mathbf{g}_{\langle k, i \rangle}^{(t)} }{\sum_{i \in \mathcal{M}_k} Q_{\langle k, i \rangle}}.
    \end{equation}}
    The PS then updates the model $\mathbf{w}_k^{(t)}$ by 
    \begin{equation}
        \label{SGD}
        \mathbf{w}_k^{(t+1)}=\mathbf{w}_k^{(t)}- \eta_k \mathbf{g}_k^{(t)},
    \end{equation}
    where $\eta_k \in \mathbb{R}$ is the learning rate of task $k$.
\end{itemize}

\subsection{\revise{MIMO MAC Channel}}
\label{MIMO_MAC_Channel}

\revise{We now introduce the wireless channel model of OA-FMTL. Following the common practice in OA-FL\cite{amiri2020federated,cao2022transmission,sery2021over}, we assume the global model is broadcast to the devices via error-free links and the local gradient uploading is synchronized among the devices \footnote{\revise{The synchronization can be realized by using the existing techniques, e.g., the timing-advance mechanism for uplink synchronization in 4G Long Term Evolution (LTE) \cite{zhu2021over}. }}. Recall that OA-FMTL is composed of a PS (i.e., a base station) and $M$ edge devices. We assume the PS and each edge device are respectively equipped with $N_{\mathrm{R}}$ and $N_{\mathrm{T}}$ antennas, which leads to a multiple-input multiple-output (MIMO) multiple access (MAC) channel.
We further assume a block-fading channel model, where the channel state keeps invariant during the step of local gradients uploading}\footnote{\revise{ We assume that perfect CSI is available at the PS. 
Studies on CSI acquisition over MIMO MAC channels can be found, e.g., in \cite{nguyen2013compressive,wen2014channel}.}}. 

We next focus on the local gradient uploading process.  
\revise{Let $C$ denote the total number of channel uses in each communication round.} Then, at the $t$-th communication round, the received signal matrix at the PS is denoted by 
\begin{equation}
    \mathbf{Y}^{(t)} = \sum\nolimits_{k=1}^K \sum\nolimits_{i \in \mathcal{M}_k} \mathbf{H}_{\langle k, i \rangle}^{(t)} \mathbf{X}_{\langle k, i \rangle}^{(t)} + \mathbf{N}^{(t)},
    \label{channelmodel}
\end{equation}
\revise{where $\mathbf{H}_{\langle k, i \rangle}^{(t)} \in \mathbb{C}^{ N_{\mathrm{R}} \times N_{\mathrm{T}}}$ denotes the channel matrix between the $\langle k, i \rangle$-th device and the PS, $\mathbf{X}_{\langle k, i \rangle}^{(t)} \triangleq \big [ \mathbf{x}_{\langle k, i \rangle}^{(t)}[1], \cdots, \mathbf{x}_{\langle k, i \rangle}^{(t)}[C] \big ] \in \mathbb{C}^{N_{\mathrm{T}} \times C}$ with $\mathbf{x}_{\langle k, i \rangle}^{(t)}[c]$ denotes the signal transmitted by the $\langle k, i \rangle$-th device at the $c$-th channel use, $\mathbf{N}^{(t)} \triangleq [\mathbf{n}^{(t)}[1], \cdots, \mathbf{n}^{(t)}[C]] \in \mathbb{C}^{N_{\mathrm{R}} \times C}$ is an additive white Gaussian noise (AWGN) matrix whose entries are independently drawn from $\mathcal{CN}(0,\sigma^2)$, and $\mathbf{Y}^{(t)} \triangleq \big[ \mathbf{y}^{(t)}[1], \cdots, \mathbf{y}^{(t)}[C] \big] \in \mathbb{C}^{ N_{\mathrm{R}} \times C}$ with $\mathbf{y}^{(t)}[c]$ denotes the received signal at the PS at the $c$-th channel use.
	
Note that the two sums in \eqref{channelmodel} imply that all the edge devices are allowed to participate in training transmit using shared radio resources. The superposition characteristic of electromagnetic waves can be utilized for signal aggregation, as suggested by AirComp, while this non-orthogonal transmission introduces inter-task interference. In the next subsection, we show how to aggregate local gradients over the air using the underlying communication channel in \eqref{channelmodel}.
}

\subsection{Over-the-Air Model Aggregation}
\revise{In this subsection, we introduce the approach of over-the-air gradient aggregation based on the underlying communication system in Section \ref{MIMO_MAC_Channel}. As illustrated in Fig.~\ref{fig:flow_figure}, the local gradients $\{\mathbf{g}_{\langle k, i \rangle}^{(t)}\mid k \in [K], i \in \mathcal{M}_k\}$ are first pre-processed to yield signals $\{\mathbf{X}_{\langle k, i \rangle}^{(t)} \mid  k \in [K], i \in \mathcal{M}_k\}$, which are then transmitted to the PS via the underlying communication channel in \eqref{channelmodel}. After the PS receives $\mathbf{Y}^{(t)}$, it separately performs post-processing to obtain $\{\hat{\mathbf{g}}_k^{(t)}\}_{k=1}^K$, which are estimates of the desired aggregated gradients $\{\mathbf{g}_k^{(t)}\}_{k=1}^K$. In the following, we introduce the pre-processing and post-processing steps.}

\begin{figure}[htbp]
\centering
\includegraphics[width=0.8\textwidth]{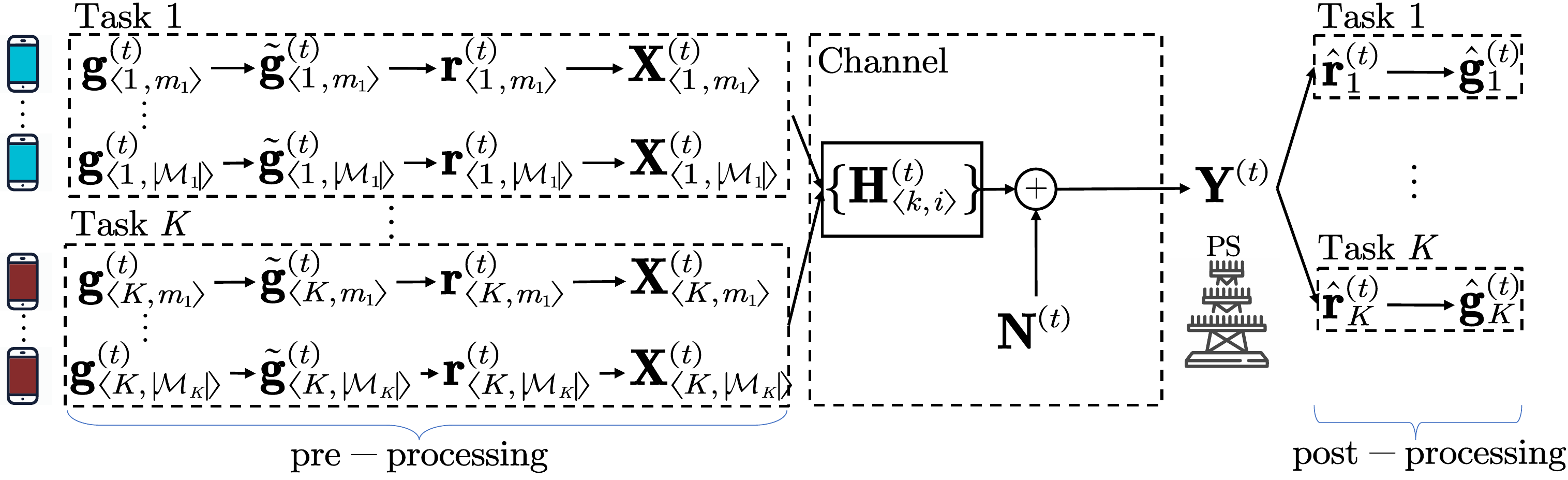}
\caption{Over-the-air gradient aggregation.}
\label{fig:flow_figure}
\end{figure}

\subsubsection{Pre-processing}
This step performs separately on each participated device, aims to generate appropriate transmitting signals. Specifically, we first normalize $\mathbf{g}_{\langle k, i \rangle}^{(t)}$ to $\mathbf{\tilde{g}}_{\langle k, i \rangle}^{(t)}\in \mathbb{R}^D$ by
\begin{equation}
    \label{nomalize}
    \tilde{\mathbf{g}}_{\langle k, i \rangle}^{(t)}[d] = \big (\mathbf{g}_{\langle k, i \rangle}^{(t)}-\bar{g}_{\langle k, i \rangle}^{(t)} \mathbf{1}_{D \times 1}\big )/\sqrt{v_{\langle k, i \rangle}^{(t)}}, 
\end{equation}
where $\bar{g}_{\langle k, i \rangle}^{(t)} = \frac{1}{D}\sum\nolimits_{d=1}^D g_{\langle k, i \rangle}^{(t)}[d]$ and $v_{\langle k, i \rangle}^{(t)} = \frac{1}{D}\sum\nolimits_{d=1}^D \left|g_{\langle k, i \rangle}^{(t)}[d] - \bar{g}_{\langle k, i \rangle}^{(t)}\right|^2$.
Following the common practice \cite{liu2020reconfigurable,lin2021deploying}, we assume $\{\bar{g}_{\langle k, i \rangle}^{(t)}, v_{\langle k, i \rangle}^{(t)} | \forall k, i, t\}$ are transmitted to the PS via error-free links. 

\revise{To match the complex communication system in \eqref{channelmodel}, we then modulate (analog modulation) the normalized gradient $\mathbf{\tilde{g}}_{\langle k, i \rangle}^{(t)}$ to a complex data vector $\mathbf{r}_{\langle k, i \rangle}^{(t)}$, i.e.,
\begin{equation}
    \mathbf{r}_{\langle k, i \rangle}^{(t)} \triangleq 
    \mathbf{\tilde{g}}_{\langle k, i \rangle}^{(t)}\left (1:\frac{D}{2} \right ) + 
    j \mathbf{\tilde{g}}_{\langle k, i \rangle}^{(t)}\left (\frac{D+2}{2} : D\right ) \in \mathbb{C}^C,
    \label{modulation}
\end{equation}
where $C = D/2$\footnote{For simplicity, we assume that $D$ is even.}.
We transmit $\mathbf{r}_{\langle k, i \rangle}^{(t)}$ with $C$ times of channel use (one element of $\mathbf{r}_{\langle k, i \rangle}^{(t)}$ for one time of channel use). Specifically, the transmitted signal is given by
\begin{equation}
\mathbf{X}_{\langle k, i \rangle}^{(t)} \triangleq \mathbf{u}_{\langle k, i \rangle}^{(t)} \mathbf{r}_{\langle k, i \rangle}^{(t)}{}^\mathrm{T} \in \mathbb{C}^{N_{\mathrm{T}} \times C},
\end{equation}
where $\mathbf{u}_{\langle k, i \rangle} \in \mathbb{C}^{N_{\mathrm{T}}}$ is the transmit beamforming vector of the $\langle k, i \rangle$-th device. Recall that $\mathbf{X}_{\langle k, i \rangle}^{(t)} = \big [ \mathbf{x}_{\langle k, i \rangle}^{(t)}[1], \cdots, \mathbf{x}_{\langle k, i \rangle}^{(t)}[C] \big ]$. Thus we have $\mathbf{x}_{\langle k, i \rangle}^{(t)}[c] = r_{\langle k, i \rangle}^{(t)}[c] \mathbf{u}_{\langle k, i \rangle}$ denoting the transmitted signal of the $\langle k, i \rangle$-th device at the $c$-th channel use, satisfying the following power constraint:
\begin{equation}
    \mathbb{E}\left[\| \mathbf{x}_{\langle k, i \rangle}^{(t)}[c]\|^2\right] = 2\|\mathbf{u}_{\langle k, i \rangle}\|^2 \leq P_0,
    \label{transbf}
\end{equation}
where the equality follows the normalization in \eqref{nomalize} (implying $\mathbb{E}\big[|r_{\langle k, i \rangle}^{(t)}[c]|^2 \big]\!=\!2$)}. 

\revise{\subsubsection{Post-processing} The received signals from the $N_{\mathrm{T}}$ receive antennas are combined separately using $K$ receive beamforming vectors to obtain $\mathbf{\hat{r}}_k^{(t)}$, i.e.,
\begin{equation}
    \mathbf{\hat{r}}_k^{(t)}\!=\!\zeta_k\!\left(\!\mathbf{f}_k^{\mathrm{H}} \mathbf{Y}_k^{(t)}\!\right){}^\mathrm{T}\!\!=\!\zeta_k\!\Big(\!
    \sum_{i \in \mathcal{M}_k}\!\mathbf{r}_{\langle k, i \rangle} \!(\mathbf{H}_{\langle k, i \rangle}^{(t)}\mathbf{u}_{\langle k, i \rangle}\!){}^\mathrm{T}
    \!+\!\sum_{l \neq k}\!\sum_{i \in \mathcal{M}_{l}}\!\mathbf{r}_{\langle l, i \rangle}^{(t)} \!(\mathbf{H}_{\langle l, i \rangle}^{(t)} \mathbf{u}_{\langle l, i \rangle}\!){}^\mathrm{T} 
    \!+\!\mathbf{N}^{(t)}{}^\mathrm{T}\!\Big) \mathbf{f}_k^{\dagger},\!k\!\in \![K],
    \label{receivedsignal}
\end{equation}
where $\zeta_k \in \mathbb{R}$ is a weighting factor and $\mathbf{f}_k \in \mathbb{C}^{N_{\mathrm{R}}}$ is the normalized receive beamforming vector of task $k$ with $\|\mathbf{f}_k\| = 1$. Estimates of the desired aggregated gradients $\{\mathbf{\hat{g}}_k^{(t)}\}_{k=1}^K$ is reconstructed from $\{\mathbf{\hat{r}}_k^{(t)}\}_{k=1}^K$ by
\begin{equation}
    \mathbf{\hat{g}}_k^{(t)} = \frac{1}{\sum_{i \in \mathcal{M}_k} Q_{\langle k, i \rangle}} \left[ 
    \Re\{\mathbf{\hat{r}}_k^{(t)}\}^\mathrm{T} , \,
    \Im\{\mathbf{\hat{r}}_k^{(t)}\}^\mathrm{T} \right]^\mathrm{T} + \bar{g}_k^{(t)}\mathbf{1}_{D \times 1} \in \mathbb{R}^D, ~k\in[K], 
    \label{receivedgradient}
\end{equation}
where $\bar{g}_k^{(t)} \triangleq (\sum_{i \in \mathcal{M}_k} Q_{\langle k, i \rangle} \bar{g}_{\langle k, i \rangle}^{(t)}) / ( \sum_{i \in \mathcal{M}_k} Q_{\langle k, i \rangle} ) $.
}

\subsection{OA-FMTL Framework}
We first describe the model aggregation error of each task $k$, and then summarize the OA-FMTL framework in this subsection. From the channel model in \eqref{channelmodel}, the model aggregations $\{\mathbf{\hat{g}}_k^{(t)}\}_{k=1}^K$ inevitably suffer from distortions caused by the inter-task interference and the channel noise.
The global model $\mathbf{w}_k^{(t)}$ of task $k$ is updated with the gradient aggregation $\mathbf{\hat{g}}_k^{(t)}$:
\begin{equation}
    \label{SGD_error_updates}
    \mathbf{w}_k^{(t+1)}=\mathbf{w}_k^{(t)}- \eta_k \mathbf{\hat{g}}_k^{(t)} = \mathbf{w}_k^{(t)}- \eta_k \left( \nabla F_k(\mathbf{w}_k^{(t)}) - \mathbf{e}_k^{(t)}\right),
\end{equation}
where $\nabla F_k(\mathbf{w}_k^{(t)}) \triangleq \frac{1}{Q_k} \sum_{n=1}^{Q_k} \nabla f_k\left( \mathbf{w}_k; \boldsymbol{\xi}_{k,n}\right) $ is the gradient of the loss function of the $k$-th task $F_k(\mathbf{w}_k)$ at $\mathbf{w}_k = \mathbf{w}_k^{(t)}$,
\revise{ and $\mathbf{e}_{k}^{(t)}$ is the error caused by gradient uploading, which can be divided into two parts: 
\begin{equation}
    \mathbf{e}_{k}^{(t)} = 
    \underbrace{\nabla\!F_k (\!\mathbf{w}_{k}^{(t)}\!)\!-\!\mathbf{g}_{k}^{(t)}}_{\mathbf{e}_{\mathrm{ds},k}^{(t)} } 
    \!+\!\underbrace{\mathbf{g}_{k}^{(t)} - \hat{\mathbf{g}}_{k}^{(t)} }_{ \mathbf{e}_{\mathrm{com},k}^{(t)} }, 
    \label{ErrorTwoPart}
\end{equation}
where $\mathbf{e}_{\mathrm{ds},k}^{(t)}$ denotes the error caused by device selection, and $\mathbf{e}_{\mathrm{com},k}^{(t)}$ denotes the communication error due to the inter-task interference and the noise. To alleviate the impact of $\mathbf{e}_{k}^{(t)}$ on the learning performance, the PS optimizes the device selection, transmit and receive beamforming.

We summarize the OA-FMTL framework described above in Algorithm \ref{alg:FL_framework}. Besides, we define the following variables for notational brevity:
$\mathcal{M} \triangleq \{\mathcal{M}_k\}_{k=1}^{K}$,
$\mathbf{f} \triangleq \{\mathbf{f}_k\}_{k=1}^K$, 
$\mathbf{u}_k \triangleq \{\mathbf{u}_{\langle k, i \rangle}\}_{i \in \mathcal{M}_k}$, and $\mathbf{u} \triangleq \{\mathbf{u}_k\}_{k=1}^{K}$.
In the following section, we establish the connection between the model aggregation errors in \eqref{ErrorTwoPart} and the learning performance of the OA-FMTL.
\begin{algorithm}[htb]
\caption{OA-FMTL framework} 
\label{alg:FL_framework} 
\begin{algorithmic}[1] 
\REQUIRE $T$, $\{Q_{\langle k, i \rangle}\}$. 
\STATE {\textbf{Initialization:}} $t = 0$, the global models $\{\mathbf{w}^{(0)}_k\}$ on the PS.
\FOR{ $t \in [T]$ }
    \STATE The PS estimates the CSI and optimize $(\mathcal{M}, \mathbf{f}, \mathbf{u})$;
    \STATE The PS sends the global models $\{\mathbf{w}^{(t)}_k\}$ to the devices through orthogonal transmission;
        \FOR{$k \in [K], i \in \mathcal{M}_k$ in parallel}
            \STATE  Device $\langle k, i \rangle$ computes its local gradient $\mathbf{g}_{\langle k, i \rangle}^{(t)}$ based on the local dataset based on \eqref{GradDef};
            \STATE  Device $\langle k, i \rangle$ uploads its gradient $\mathbf{g}_{\langle k, i \rangle}^{(t)}$ to the PS via \eqref{nomalize}-\eqref{transbf};
        \ENDFOR
    \STATE The PS recovers the aggregated gradient $\mathbf{\hat{g}}_k^{(t)}$ of each task $k$ based on \eqref{receivedsignal} and \eqref{receivedgradient};
    \STATE The PS updates the global models $\{\mathbf{w}^{(t+1)}_k\}$ based on \eqref{SGD_error_updates};
\ENDFOR 
\end{algorithmic}
\end{algorithm}
}

\section{Performance Analysis}
In this section, we analyse the learning performance of the OA-FMTL. We start with some standard assumptions on the loss functions $\{F_k(\cdot)\}_{k=1}^K$,

\subsection{Preliminaries}
To conduct convergence analysis, following the stochastic optimization literature \cite{friedlander2012hybrid, bertsekas1995neuro}, we make the following four assumptions on $\{F_k(\cdot)\}_{k=1}^K$:
\begin{assumption}
    \label{asp:f_Lipschitz}
    For each task $k$, the loss function $F_k$ is continuously differentiable, and the gradient $\nabla F_k(\cdot)$ is uniformly Lipschitz continuous with parameter $\omega_k$, i.e., 
    \begin{equation}
        \label{Lipschitz}
        \| \nabla F_k(\mathbf{w}) - \nabla F_k(\mathbf{w^\prime})\| \leq \omega_k \|\mathbf{w} - \mathbf{w^\prime} \|, \forall \mathbf{w}, \mathbf{w^\prime} \in \mathbb{R}^D, k \in [K].
    \end{equation} 
\end{assumption}
\begin{assumption}
    \label{asp:f_convex}
    For each task $k$, loss function $F_k$ is strongly convex with positive parameter $\mu_k$:
    \begin{equation}
        F_k(\mathbf{w}) 
        \geq F_k(\mathbf{w^\prime}) + (\mathbf{w} - \mathbf{w^\prime})^\mathrm{T} \nabla F_k(\mathbf{w^\prime}) + \frac{\mu_k}{2} \| \mathbf{w} - \mathbf{w^\prime}\|^2, \forall \mathbf{w}, \mathbf{w^\prime} \in \mathbb{R}^D, k \in [K].
    \end{equation}
\end{assumption}
\begin{assumption}
    \label{asp:f_twice_diff}
    $F_k(\cdot), k \in [K]$ are twice-continuously differentiable.
\end{assumption}
\begin{assumption}
    \label{asp:f_bound}
    The gradient vector is upper bounded by
    \begin{equation}
        \label{eq:f_bound}
        \begin{aligned}
            \| \nabla f_k\left(\mathbf{w}_k; \boldsymbol{\xi}_{k,n}\right) \|^2 \leq \beta_1 + \beta_2 \| \nabla F_k(\mathbf{w}_k^{(t)}) \|^2, \forall k \in [K],\\
        \end{aligned}
    \end{equation}
    for some constants $\beta_{1} \geq 0$ and $\beta_{2} \geq 0$. Both $\beta_{1}$ and $\beta_{2} $ are constants shared by $\{F_k(\cdot)\}_{k=1}^K$.
\end{assumption}
\begin{figure}[htbp]
    \centering
    \includegraphics[width=0.9\textwidth]{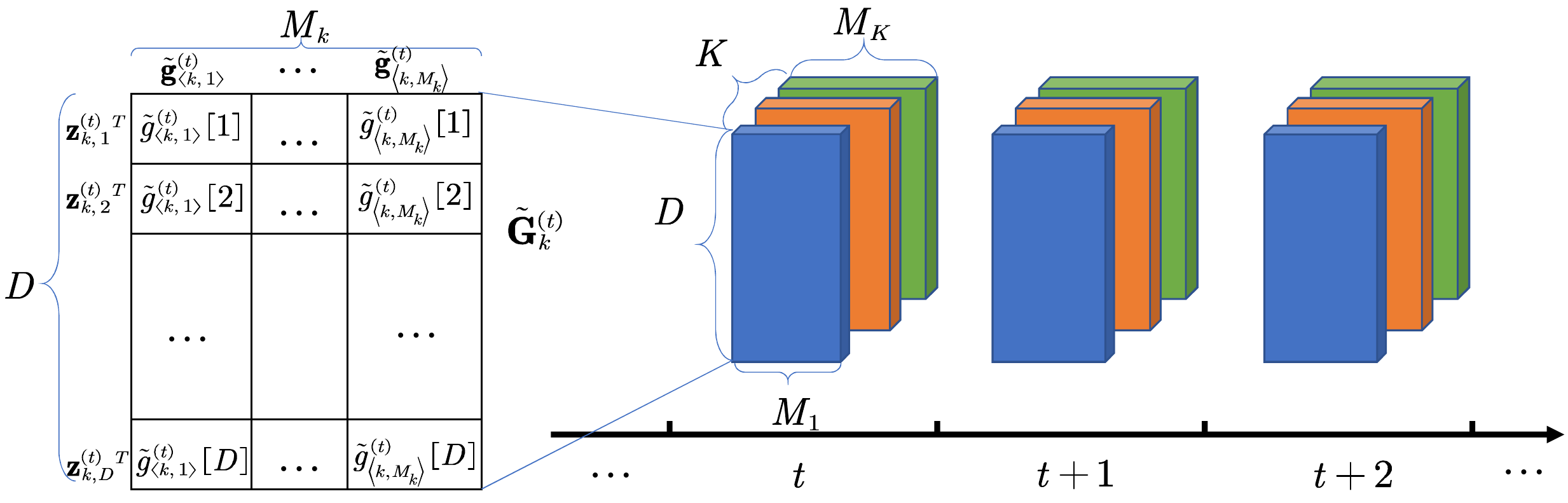}
    \caption{An illustration of the gradients in the OA-FMTL}
    \label{fig:grad_mtrx}
\end{figure}
Furthermore, unlike the approaches in \cite{cao2021cooperative, yang2020federated} where the local gradients from a common task are treated as independent, we find that these local gradients are highly correlated in a typical learning task. As such, it is of critical importance to understand the impact of the spatial correlation between gradients on the learning performance. To this end, we introduce a probability model for the gradients as follows. Let $\mathbf{\tilde{G}}_k^{(t)} \triangleq [\mathbf{\tilde{g}}_{\langle k,1 \rangle}^{(t)}, \cdots, \mathbf{\tilde{g}}_{\langle k,M_k \rangle}^{(t)}] \in \mathbb{R}^{D \times M_k}$ be the local gradients from the $M_k$ devices of task $k$ at the $t$-th round. Let $\mathbf{z}_{k,d}^{(t)} \in \mathbb{R}^{M_k}$ be the $d$-th dimension of the local gradients from a common task, and $\mathbf{z}_{k,d}^{(t)}{}^\mathrm{T}$ is the $d$-th row of $\mathbf{\tilde{G}}_k^{(t)}$, as shown in Fig.~\ref{fig:grad_mtrx}.
To track the correlation of the gradients, we make the following assumption on the distribution of the gradients elements.
\begin{assumption}
    \label{asp:g_task}
    For the $t$-th communication round, the gradient matrices $\{\mathbf{\tilde{G}}_k^{(t)} | k \in [K]\}$ are independent and non-identically distributed. For the $k$-th task, the gradients $\{\mathbf{z}_{k,d}^{(t)} | d \in [D]\}$, are independent and identically distributed. That is,
    \begin{equation}
        p^{(t)} \left(\{\mathbf{\tilde{G}}_k^{(t)} | k \in [K]\} \right) = \prod_{k=1}^K \prod_{d=1}^D p_k^{(t)} \left( \mathbf{z}_{k,d}^{(t)} \right), \forall d,
    \end{equation}
    where $p^{(t)}(\cdot)$ denotes the distribution of the elements of $\{\mathbf{\tilde{G}}_k^{(t)} | k \in [K]\}$, and $p_k^{(t)}(\cdot)$ denotes the distribution of the elements of $\mathbf{z}_{k,d}^{(t)}$. Furthermore, for each task $k$, the local gradients of the $M_k$ devices have the same degree of variation, i.e., $v_{\langle k,1 \rangle}^{(t)} = \cdots = v_{\langle k,M_k \rangle}^{(t)} = v_k^{(t)} $.
\end{assumption}

We now focus on the spatial correlation between the local gradients in a common task, i.e., between the entries of $\mathbf{z}_{k,d}^{(t)}$. 
The auto-correlation matrix for task $k$ is then defined by
\begin{equation}
    \label{rho}
    \boldsymbol{\rho}_k^{(t)} \triangleq \mathbb{E}\big[\mathbf{z}_{k,d}^{(t)} (\mathbf{z}_{k,d}^{(t)})^\mathrm{T} \big] \in \mathbb{R}^{M_k\times M_k}, \forall d
\end{equation}
where the $(i,j)$-th entry is denoted as $\rho_{\langle k, i \rangle,\langle k, j \rangle}^{(t)}$, measuring the spatial correlation between device $i$ and device $j$ of task $k$. Note that each element of $\mathbf{z}_{k,d}^{(t)}$, $\tilde{g}_{\langle k, i \rangle}^{(t)}[d]$, has zero mean and unit variance. Thus, $\rho_{\langle k, i \rangle,\langle k, i \rangle}^{(t)} = 1$, and $\rho_{\langle k, i \rangle,\langle k, j \rangle}^{(t)} \in [-1, 1]$, for $\forall k,i,j$. 
\begin{figure}[htbp]
    \centering
    \includegraphics[width=0.63\textwidth]{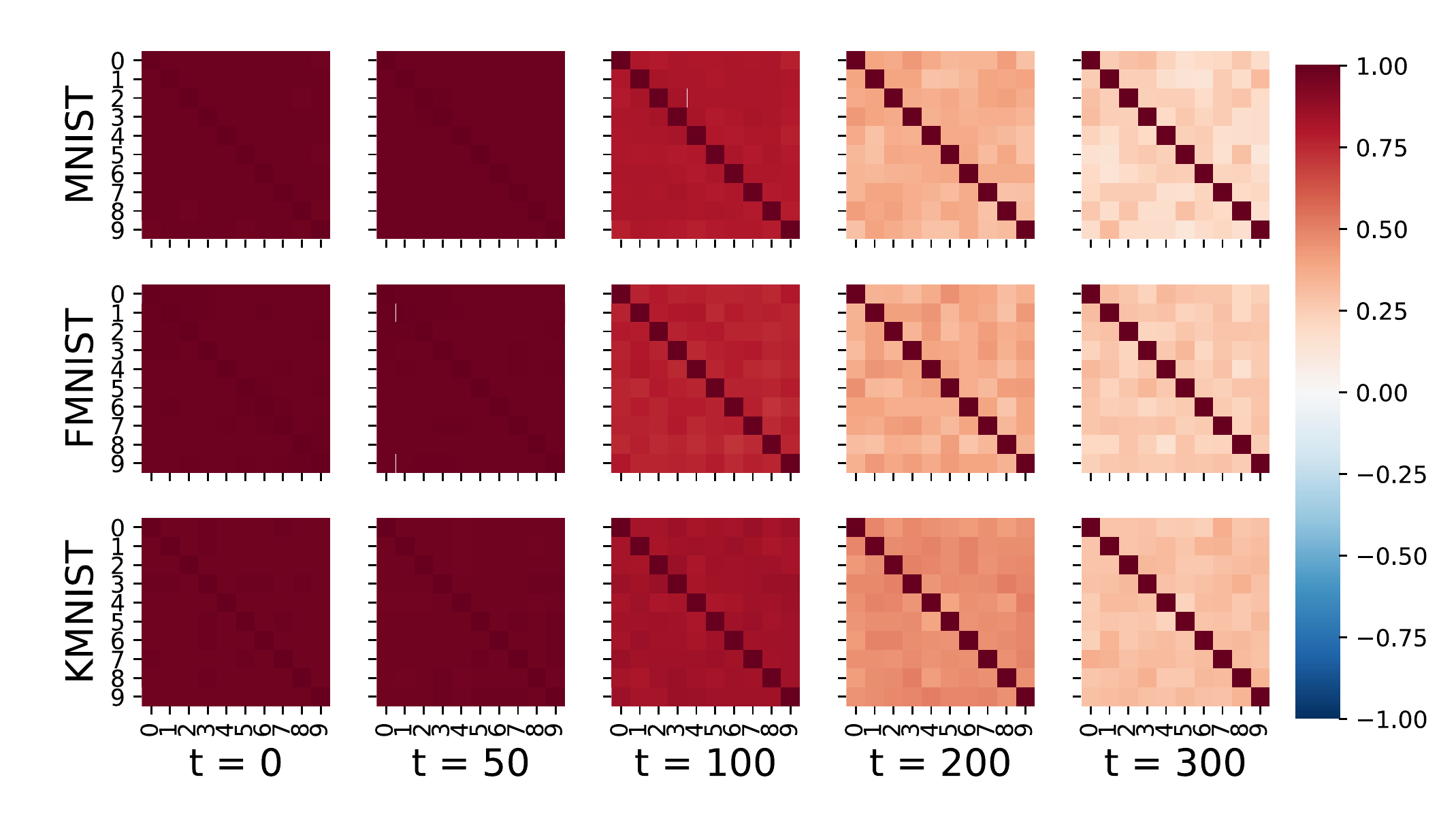}
    \caption{Experimental results of $\boldsymbol{\rho}_k^{(t)}$ versus communication rounds $t$. We train three LeNet\cite{lecun1998gradient}-based FL models, with $M_1 = M_2 = M_3 = 10$ and $20$ Monte Carlo trials. The learning rate is set to $\eta_1 = \eta_2 = \eta_3 = 0.002$, and the momentum is set to $0.9$. The loss function is the cross-entropy loss. Each training dataset is assigned $60000$ samples and the local training data is assigned $6000$ samples i.i.d. drawn from the dataset. The local updates have $5$ times of stochastic gradient descents (SGD), and the mini-batch size is set to be $1200$. The gradients uploading is error-free and all devices are selected. We approximate $\rho_{\langle k, i \rangle,\langle k, j \rangle}^{(t)}$ by $ \rho_{\langle k, i \rangle,\langle k, j \rangle}^{(t)} \approx \frac{1}{D} \sum_{d=1}^D \tilde{g}_{\langle k, i \rangle}^{(t)}[d] \tilde{g}_{\langle k, j \rangle}^{(t)}[d] $. }
    \label{fig:rho}
\end{figure}
The heatmaps in Fig.~\ref{fig:rho} illustrate the experimental results of $\boldsymbol{\rho}_k^{(t)}$ versus the communication round $t$ of three datasets, MNIST\cite{lecun2010mnist}, Fashion-MNIST\cite{xiao2017fashion}, and KMNIST\cite{clanuwat2018deep}, where the gradients are updated ideally without any transmission error. Intuitively, the darker the color of a pixel, the stronger the spatial correlation between the gradients from the corresponding two devices. 
From Fig.~\ref{fig:rho}, we see that the correlation grows substantially at the beginning of training, say, for communication round $t\leq 50$. When the training approaches convergence, most of the elements of $\boldsymbol{\rho}_k^{(t)}$ reduce to less than 0.3, which indicates that the local gradients have weaker cross-correlation when the models are close to convergence. 

\subsection{Convergence Analysis of OA-FMTL}
We start with the analysis of the loss function $F_k(\cdot)$ of task $k$ at each communication round. From [\citenum{friedlander2012hybrid}, Lemma~2.1], Assumptions~\ref{asp:f_Lipschitz}-\ref{asp:f_bound} lead to an upper bound of the loss function $F_k(\cdot)$ at the $t$-th round of the recursive updates in \eqref{SGD_error_updates}. 
We therefore have the following theorem.
\begin{theorem}
    \label{Theorem:MSE}
    Under Assumptions~\ref{asp:f_Lipschitz}-\ref{asp:g_task}, the expected loss function of each task $k$ at the $t$-th communication round is bounded by
    \begin{equation}
        \label{F_Upperbound}
        \begin{aligned}
            \mathbb{E}[F_k(\mathbf{w}_k^{(t+1)})]\!\leq\!
            \mathbb{E}[F_k(\mathbf{w}_k^{(t)})]\!- \frac{1}{ 2\omega_k}
            \left(\!\mathbb{E}[\|\nabla F_k(\mathbf{w}_k^{(t)})\|^{2}]  \!-\!2\mathbb{E} [ \|\mathbf{e}_{\mathrm{ds},k}^{(t)} \|^{2}]\!-\! 2\mathbb{E}[\|\mathbf{e}_{\mathrm{com},k}^{(t)}\|^{2}]\right),
        \end{aligned}
    \end{equation}
    where the device selection MSE is bounded by
    \begin{equation}
        \label{ErrorDevicesSelection}
        \mathbb{E}[\|\mathbf{e}_{\mathrm{ds},k}^{(t)} \|^{2}] \leq \frac{4}{Q_k^{2}}\Big(Q_k-\sum\nolimits_{i \in \mathcal{M}_k} Q_{\langle k, i \rangle}\Big)^{2}\left(\beta_1 + \beta_2 \mathbb{E}\left[ \|\nabla F_k(\mathbf{w}_k^{(t)}) \|^{2} \right] \right),
    \end{equation}
    and the communication MSE is given by
    \begin{align}
        \label{ErrorCommunication}
        \mathbb{E}[\|\mathbf{e}_{\mathrm{com},k}^{(t)}\|^{2}]
        = & \frac{1}{(\sum_{i \in \mathcal{M}_k} Q_{\langle k, i \rangle})^2} \sum\nolimits_{c=1}^{C}\!\!\bigg ( 
        \underbrace{\mathbb{E}
        \bigg[\Big|\sum\nolimits_{i \in \mathcal{M}_k}\!\Big(Q_{\langle k, i \rangle} \sqrt{v_k^{(t)}}\!-\!\zeta_k \mathbf{f}_k^\mathrm{H} \mathbf{H}_{\langle k, i \rangle}^{(t)} \mathbf{u}_{\langle k, i \rangle}\Big) r_{\langle k, i \rangle}^{(t)}[c]  \Big|^2\bigg]}_{\mathrm{the\;first\;term:\;the\;misalignment\;error}}
        \notag \\
        & + 
        \underbrace{\zeta_k^2 \sum\nolimits_{l \neq k} \mathbb{E} \Big[\Big|\sum\nolimits_{i \in \mathcal{M}_{l}} \mathbf{f}_k^\mathrm{H} \mathbf{H}_{\langle l, i \rangle}^{(t)} \mathbf{u}_{\langle l, i \rangle} r_{\langle l, i \rangle}^{(t)}[c] \Big|^{2}\Big]}_{\mathrm{the\;second\;term:\;the\;interference}} + 
        \underbrace{\zeta_k^2 \mathbb{E}\Big[\Big|\mathbf{f}_k^\mathrm{H} \mathbf{n}^{(t)}[c] \Big|^{2}\Big]}_{\mathrm{the\;third\;term:\;the\;noise}} 
        \bigg ).
    \end{align}
\end{theorem}
\begin{proof}
    Please refer to Appendix \ref{Appendix:Prf_Theorem:MSE}.
\end{proof}
From Theorem~\ref{Theorem:MSE}, we obtain an upper bound of the loss function $F_k(\cdot)$ w.r.t. the device selection MSE and the communication MSE.
By inspection, the communication MSE in \eqref{ErrorCommunication} consists of three terms, 
where the first term represents the misalignment error of the aggregation gradients from the devices in task $k$, 
the second term represents the error caused by the interference from the devices associated with other tasks, 
and the third term represents the error caused by the channel noise.
Note that the expression in \eqref{ErrorCommunication} is convex w.r.t. $\zeta_k$ for any fixed device selection set $\mathcal{M}_k$ and beamforming $\mathbf{f}_k$ and $\{\mathbf{u}_{\langle k, i \rangle}\}_{i \in \mathcal{M}_k}$. Thus, we have the following corollary.
\begin{corollary}
    \label{Corollary:MSE}
    The optimal $\zeta_k$ is given by 
    \begin{equation}
        \label{opt_c}
        \zeta_k^* = \frac{ \sqrt{v_k^{(t)}} \sum_{i,j \in \mathcal{M}_k} \rho_{\langle k, i \rangle, \langle k, j \rangle}^{(t)} \left(Q_{\langle k, i \rangle} (\mathbf{f}_k^\mathrm{H} \mathbf{H}_{\langle k, j \rangle}^{(t)} \mathbf{u}_{\langle k, j \rangle})^\mathrm{H} + Q_{\langle k, j \rangle} \mathbf{f}_k^\mathrm{H} \mathbf{H}_{\langle k, i \rangle}^{(t)} \mathbf{u}_{\langle k, i \rangle}\right)  }{2 \left(\sum_{l=1}^K \sum_{i,j \in \mathcal{M}_l} \rho_{\langle l, i \rangle, \langle l, j \rangle}^{(t)} (\mathbf{f}_k^\mathrm{H} \mathbf{H}_{\langle l, i \rangle}^{(t)} \mathbf{u}_{\langle l, i \rangle})^\mathrm{H} \mathbf{f}_k^\mathrm{H} \mathbf{H}_{\langle l, j \rangle}^{(t)}\mathbf{u}_{\langle l, j \rangle} +\sigma^2  \| \mathbf{f}_k \|^2/2\right)}.
    \end{equation}
\end{corollary} 
\begin{proof}
    Please refer to Appendix \ref{Appendix:Prf_Corollary:MSE}.
\end{proof}
\label{sec:zeta_select}
We emphasize that the design strategy of the weighting factors $\{\zeta_k\}_{k=1}^K$ in \eqref{opt_c} is different from that of the existing over-the-air FL approaches. For each task $k$, the first component in \eqref{ErrorCommunication} represents the misalignment error of the gradient aggregation. In the existing scheme, Refs.~\cite{liu2020reconfigurable, yang2020federated, shi2021over} force this component to zero, which leads to the constraints of $Q_{\langle k, i \rangle} \sqrt{v_k^{(t)}}-\zeta_k \mathbf{f}_k^\mathrm{H} \mathbf{H}_{\langle k, i \rangle}^{(t)} \mathbf{u}_{\langle k, i \rangle} = 0, \forall i \in \mathcal{M}_k$, and then minimize the rest of the communication MSE to determine $\zeta_k$. For all devices in the $k$-th task, to satisfy the constraints and the transmit power budgets in \eqref{transbf}, the choice of $\zeta_k$ is therefore given by $\zeta_k^2 = \max_{i \in \mathcal{M}_k} Q_{\langle k, i \rangle}^2 v_k^{(t)} / (P_0 \|\mathbf{f}_k^\mathrm{H} \mathbf{H}_{\langle k, i \rangle}^{(t)} \|^2)$. Thus, $\zeta_k$ is dominated by the device with the worst channel condition (in terms of $\|\mathbf{f}_k^\mathrm{H} \mathbf{H}_{\langle k, i \rangle}^{(t)}\|^2$). That is, the device with the worst channel becomes the bottleneck of the overall scheme, also known as the straggler problem. 
However, instead of zero-forcing, our approach tolerances the misalignment error but requires the design of $\zeta_k$ to directly minimize the overall communication MSE for task $k$. In this way, $\zeta_k$ is no longer solely determined by the worst channel condition, which significantly relieves the straggler problem. \revise{Note that our approach is also different from the approach in \cite{cao2021cooperative} since we aware the spatial correlation between the local gradients, which brings substantial improvement in learning performance.} Numerical results will be presented later in Section~V for verification.

In the former of this subsection, we obtain an upper bound of the loss function $F_k(\cdot)$ of task $k$ at the $t$-th communication round in Theorem~\ref{Theorem:MSE}. In the following, we analyse the convergence performance of the entire OA-FMTL model at the $t$-th round, and obtain an upper bound of the average difference between the overall loss function at the $(t+1)$-th round $\mathcal{F}\left(\mathbf{w}^{(t+1)}\right)$ and the optimal $\mathcal{F}\left(\mathbf{w}^*\right)$, i.e., $\mathbb{E}\left[\mathcal{F}\left(\mathbf{w}^{(t+1)}\right)- \mathcal{F}\left(\mathbf{w}^*\right)\right]$.
Based on Theorem~\ref{Theorem:MSE} and Corollary~\ref{Corollary:MSE}, we bring the device selection MSE and the communication MSE into an analysis framework, and obtain the following theorem.

\begin{theorem}
    \label{Theorem:Convergence}
    Based on Assumptions~\ref{asp:f_Lipschitz}-\ref{asp:f_bound}, the overall loss function at the $t$-th round $\mathcal{F}(\mathbf{w}^{(t+1)})$ satisfies the following inequality:
    \begin{align}
        \label{MSERecursionUpperbound}
        \mathbb{E}[\mathcal{F}(\mathbf{w}^{(t+1)}) \!-\! \mathcal{F}(\mathbf{w}^*)]
        & \leq \mathbb{E}[\mathcal{F}(\mathbf{w}^{(t)})\!-\! \mathcal{F}(\mathbf{w}^*)]
        \Big (1 - \frac{\mu}{\omega} \Big ) \notag\\
        & \quad + \Big (\frac{2\mu\beta_2}{\omega}\mathbb{E}[\mathcal{F}(\mathbf{w}^{(t)})\!-\! \mathcal{F}(\mathbf{w}^*)] + \frac{\beta_1}{\omega} \Big )
        \mathcal{E}^{(t)}(\mathcal{M}, \mathbf{f}, \mathbf{u}),
    \end{align}
    where $\omega \triangleq \max_k \omega_k$, $\mu \triangleq \min_k \mu_k$, 
    $\mathcal{E}^{(t)} (\mathcal{M}, \mathbf{f}, \mathbf{u}) \triangleq \sum_{k=1}^K d_k^{(t)}(\mathcal{M}_k, \mathbf{f}_k, \mathbf{u}_k)$, 
    and $d_k^{(t)} (\mathcal{M}_k, \mathbf{f}_k, \mathbf{u}_k)$ is denoted by
    \begin{align}
        d_k^{(t)} & (\mathcal{M}_k, \mathbf{f}_k, \mathbf{u}_k)
        \triangleq 
        \frac{4}{Q_k^2}\!\left(Q_k\!-\!\sum_{i \in \mathcal{M}_k} Q_{\langle k, i \rangle}\right)^{2}\!+\! \frac{1}{\left(\sum_{i \in \mathcal{M}_k} Q_{\langle k, i \rangle}\right)^2} \left ( \sum_{i,j \in \mathcal{M}_k} \rho_{\langle k, i \rangle, \langle k, j \rangle}^{(t)}  Q_{\langle k, i \rangle}Q_{\langle k, j \rangle} \right. \notag\\
        & \left.- \frac{ \left( \sum_{i,j \in \mathcal{M}_k} \rho_{\langle k, i \rangle, \langle k, j \rangle}^{(t)} \left(Q_{\langle k, i \rangle} (\mathbf{f}_k^\mathrm{H} \mathbf{H}_{\langle k, j \rangle}^{(t)} \mathbf{u}_{\langle k, j \rangle})^\mathrm{H} + Q_{\langle k, j \rangle} \mathbf{f}_k^\mathrm{H} \mathbf{H}_{\langle k, i \rangle}^{(t)} \mathbf{u}_{\langle k, i \rangle}\right) \right) ^2 }{4 \left(\sum_{l=1}^K \sum_{i,j \in \mathcal{M}_{l}} \rho_{\langle l, i \rangle, \langle l, j \rangle}^{(t)} (\mathbf{f}_k^\mathrm{H} \mathbf{H}_{\langle l, i \rangle}^{(t)} \mathbf{u}_{\langle l, i \rangle})^\mathrm{H} \mathbf{f}_k^\mathrm{H} \mathbf{H}_{\langle l, j \rangle}^{(t)} \mathbf{u}_{\langle l, j \rangle} + \sigma^2 \| \mathbf{f}_k \|^2 /2 \right) } \right ).
        \label{d}
    \end{align}
\end{theorem}
\begin{proof}
    Please refer to Appendix \ref{Appendix:Prf_Theorem:Convergence}.
\end{proof}


\noindent Theorem~\ref{Theorem:Convergence} provides a metric for evaluating the learning performance of the OA-FMTL model. In the next section, we formulate the optimization problem based on this metric, and propose an efficient algorithm to solve this problem.

\section{System Optimization}
To achieve better learning performance, we aim to minimize the upper bound of $\mathbb{E}[\mathcal{F} \left(\mathbf{w}^{(t+1)}\right) -  \mathcal{F}\left(\mathbf{w}^*\right)]$ in \eqref{MSERecursionUpperbound}, or equivalently, to minimize $\mathcal{E}^{(t)}(\mathcal{M}, \mathbf{f}, \mathbf{u})$ over device selection set $\mathcal{M}$, receive beamforming $\mathbf{f}$ and transmit beamforming $\mathbf{u}$, as detailed below. 

\subsection{Problem Formulation}
\revise{ From Theorem \ref{Theorem:Convergence}, the gap $\mathbb{E}[\mathcal{F}(\mathbf{w}^{(t+1)}) \!-\! \mathcal{F}(\mathbf{w}^*)]$ has an upper bound in \eqref{MSERecursionUpperbound}. The upper bound is monotonically increasing w.r.t. $\mathcal{E}^{(t)}(\mathcal{M}, \mathbf{f}, \mathbf{u})$ since $\frac{2\mu\beta_2}{\omega}\mathbb{E}[\mathcal{F}(\mathbf{w}^{(t)})\!-\! \mathcal{F}(\mathbf{w}^*)] + \frac{\beta_1}{\omega} > 0$. With the target to minimize the gap $\mathbb{E}[\mathcal{F}(\mathbf{w}^{(t+1)}) \!-\! \mathcal{F}(\mathbf{w}^*)]$, we minimize $\mathcal{E}^{(t)}(\mathcal{M}, \mathbf{f}, \mathbf{u})$ at round $t$, and we formulate the optimization problem P1 as:}
\begin{mini!}|s|
    {\substack{\mathcal{M}, \mathbf{f}, \mathbf{u}}}
    {\mathcal{E}^{(t)}(\mathcal{M}, \mathbf{f}, \mathbf{u}) = \sum\nolimits_{k=1}^K d_k^{(t)}(\mathcal{M}_k, \mathbf{f}_k, \mathbf{u}_k) }
    {}
    {\text{(P1)}:}
    \addConstraint{\mathcal{M}_k \subset  [M_k], k \in [K] 
        \label{constraint_x}}
    \addConstraint{\left\| \mathbf{f}_k \right\| = 1, k \in [K]
        \label{constraint_f}}
    \addConstraint{\left\| \mathbf{u}_{\langle k, i \rangle} \right\|^2 \leq P_0/2, k \in [K], i \in [M_k],                         \label{constraint_u}}
\end{mini!}
where 
\begin{align}
    d_k^{(t)}(\mathcal{M}_k, \mathbf{f}_k, \mathbf{u}_k) =
    & \frac{4}{Q_k^2}\!\left(Q_k\!-\!\sum_{i \in \mathcal{M}_k} Q_{\langle k, i \rangle}\right)^{2}
    \!+\!\frac{\sum_{i,j \in \mathcal{M}_k}\!\rho_{\langle k, i \rangle, \langle k, j \rangle}^{(t)}  Q_{\langle k, i \rangle}Q_{\langle k, j \rangle}\!-\!\frac{ a_k^{(t)}(\mathcal{M}_k, \mathbf{f}_k, \mathbf{u}_k) ^2 }{4 b_k^{(t)}(\mathcal{M}_k, \mathbf{f}_k, \mathbf{u}_k)} }{\left(\sum_{i \in \mathcal{M}_k} Q_{\langle k, i \rangle}\right)^2}\!,
    \label{d_k_p1}
\end{align}
with
\begin{subequations}
    \begin{align}
        & a_k^{(t)} (\mathcal{M}_k,\!\mathbf{f}_k,\!\mathbf{u}_k)\!\triangleq\! \sum\nolimits_{i,j \in \mathcal{M}_k} \rho_{\langle k, i \rangle, \langle k, j \rangle}^{(t)} \left(Q_{\langle k, i \rangle} (\mathbf{f}_k^\mathrm{H} \mathbf{H}_{\langle k, j \rangle}^{(t)} \mathbf{u}_{\langle k, j \rangle})^\mathrm{H} \!+\!Q_{\langle k, j \rangle} \mathbf{f}_k^\mathrm{H} \mathbf{H}_{\langle k, i \rangle}^{(t)} \mathbf{u}_{\langle k, i \rangle}\right),\!\\
        & b_k^{(t)} (\mathcal{M}_k,\!\mathbf{f}_k,\!\mathbf{u}_k)\!\triangleq\! \sum\nolimits_{l=1}^K \sum\nolimits_{i,j \in \mathcal{M}_l} \rho_{\langle l, i \rangle, \langle l, j \rangle}^{(t)} (\mathbf{f}_k^\mathrm{H} \mathbf{H}_{\langle l, i \rangle}^{(t)} \mathbf{u}_{\langle l, i \rangle})^\mathrm{H} \mathbf{f}_k^\mathrm{H} \mathbf{H}_{\langle l, j \rangle}^{(t)} \mathbf{u}_{\langle l, j \rangle} + \sigma^2 \| \mathbf{f}_k \|^2 /2.\!
    \end{align}
\end{subequations}
P1 is an optimization problem w.r.t. device selection set $\mathcal{M}$, receive beamforming vectors $\mathbf{f}$ and transmit beamforming vectors $\mathbf{u}$, respectively. 

We next design an AO-based algorithm to solve the optimization problem P1. P1 contains $2M + K$ optimization variables, i.e., $M$ device selection indices, $K$ receive beamforming vectors at the PS and $M$ transmit beamforming vectors at the  devices. P1 is non-convex due to the coupling of $\mathcal{M}$, $\mathbf{f}$ and $\mathbf{u}$. Thus, we adopt the AO framework to solve the problem in a suboptimal fashion. 
First, we optimize the beamforming vectors with fixed device selection set $\mathcal{M}$. Second, with fixed beamforming vectors, we optimize device selection set $\mathcal{M}$ with Gibbs sampling\cite{liu2020reconfigurable}. The two steps iterate until convergence. The details are discussed in the following subsections.

\subsection{ Optimization of \texorpdfstring{$\mathbf{f}$}{Lg} and \texorpdfstring{$\mathbf{u}$}{Lg} with fixed \texorpdfstring{$\mathcal{M}$}{Lg}} 
We first optimize beamforming vectors $\mathbf{f}$ and $\mathbf{u}$ with fixed device selection set $\mathcal{M}$. 
\revise{ By inspection of (P1), $d_k^{(t)}(\mathcal{M}_k, \mathbf{f}_k, \mathbf{u}_k)$ in \eqref{d_k_p1} is invariant to the value of $\|\mathbf{f}_k\|$. Therefore, the unit-length constraint of $\mathbf{f}_k$ in \eqref{constraint_f} can be ignored without changing the minimum of (P1). } Further, we drop the constant terms in the objective function $\mathcal{E}^{(t)}(\mathcal{M}, \mathbf{f}, \mathbf{u})$ to obtain 
\begin{equation}
    \min_{\{\mathbf{f}_k\}, \mathbf{u}}
    \quad -\sum_{k = 1}^K \frac{ a_k^{(t)}(\mathcal{M}_k, \mathbf{f}_k, \mathbf{u}_k)^2 }{4 \left(\sum_{i \in \mathcal{M}_k} Q_{\langle k, i \rangle}\right)^2  b_k^{(t)}(\mathcal{M}_k, \mathbf{f}_k, \mathbf{u}_k)}, 
    \; \mathrm{s.t.} \quad \eqref{constraint_u}.
    \label{P2obj}
\end{equation}
\eqref{P2obj} is still non-convex because both the numerator and the denominator in the $k$-th summand contain the optimization variables $\mathbf{f}_k$ and $\mathbf{u}_k$. We adopt the quadratic transform in fractional programming (FP) \cite{shen2018fractional} to decouple the numerator and the denominator as
\begin{equation}
    \min_{\{\mathbf{f}_k\}, \mathbf{u}, \mathbf{y}}
    \quad -\sum_{k = 1}^K  \bigg( \frac{y_k a_k^{(t)}(\mathcal{M}_k, \mathbf{f}_k, \mathbf{u}_k)}{ \sum_{i \in \mathcal{M}_k} Q_{\langle k, i \rangle} } - y_k^2 b_k^{(t)}(\mathcal{M}_k, \mathbf{f}_k, \mathbf{u}_k) \bigg), 
    \; \mathrm{s.t.} \quad \eqref{constraint_u}.
    \label{P3obj}
\end{equation}
where $\mathbf{y} = [y_{1}, \cdots, y_k]^\mathrm{T} \in \mathbb{R}^K$ is an auxiliary vector introduced by FP. Note that \eqref{P3obj} reduces to \eqref{P2obj} by letting each $y_k$ take its optimal form as
\begin{equation}
    \label{y_define}
    y_k = \frac{a_k^{(t)}(\mathcal{M}_k, \mathbf{f}_k, \mathbf{u}_k)}{2 \left(\sum_{i \in \mathcal{M}_k} Q_{\langle k, i \rangle}\right) b_k^{(t)}(\mathcal{M}_k, \mathbf{f}_k, \mathbf{u}_k) }.
\end{equation}

\subsubsection{Optimizing \texorpdfstring{$\mathbf{u}_{\langle k, i \rangle}$}{Lg} with fixed \texorpdfstring{$\{\mathbf{u}_{\langle k, i \rangle}\}_{j\neq i}$}{Lg} and \texorpdfstring{$\{\mathbf{f}_k\}$}{Lg}}
\label{opt_u}
With fixed $\{\mathbf{u}_{\langle k, i \rangle}\}_{j\neq i}$ and $\{\mathbf{f}_k\}$, when $i\!\notin\!\mathcal{M}_k$, i.e., the device $\langle k, i \rangle$ is not selected, we obtain $\mathbf{u}_{\langle k, i \rangle} = \mathbf{0}_{N_{\mathrm{T}} \times 1} $ directly. When $i\!\in\!\mathcal{M}_k$, i.e., the device $\langle k, i \rangle$ is selected, \eqref{P3obj} reduces to 
\begin{equation*}
    \text{(P2)}: \quad \min_{\mathbf{u}_{\langle k, i \rangle}}
    \quad \mathbf{u}_{\langle k, i \rangle}^\mathrm{H} \mathbf{A}_{\langle k, i \rangle}^{(t)} \mathbf{u}_{\langle k, i \rangle}  -2 \Re \left\{ \mathbf{b}_{\langle k, i \rangle}^{(t)}{}^\mathrm{H} \mathbf{u}_{\langle k, i \rangle} \right \}, 
    \; \mathrm{s.t.} \quad \eqref{constraint_u}.
\end{equation*}
where $\mathbf{A}_{\langle k, i \rangle}^{(t)} \in \mathbb{C}^{N_{\mathrm{T}} \times N_{\mathrm{T}}}$ and $\mathbf{b}_{\langle k, i \rangle}^{(t)} \in \mathbb{C}^{N_{\mathrm{T}}}$ are defined by
\begin{subequations}
    \label{Ab_device}
    \begin{align}
        \mathbf{A}_{\langle k, i \rangle}^{(t)}
        &\triangleq \sum\nolimits_{l = 1}^K y_{l}^2 (\mathbf{H}_{\langle k, i \rangle}^{(t)})^\mathrm{H} \mathbf{f}_l \mathbf{f}_l^\mathrm{H} \mathbf{H}_{\langle k, i \rangle}^{(t)},\\
        \mathbf{b}_{\langle k, i \rangle}^{(t)} 
        &\!\triangleq\!\frac{y_k\!\sum_{j\!\in \mathcal{M}_k}\!\rho_{\langle k, i \rangle, \langle k, j \rangle}^{(t)}\!Q_{\langle k, j \rangle}\! (\!\mathbf{f}_k^\mathrm{H}\!\mathbf{H}_{\langle k, i \rangle}^{(t)}\!)^\mathrm{H} }{\sum_{i \in \mathcal{M}_{\!k}} Q_{\langle k, i \rangle} }
        \!-\!\sum_{\!l = 1}^K\!y_{l}^2\!\sum_{\!j\!\in \mathcal{M}_{\!k}, j\!\neq\!i} \!\rho_{\langle\!k, i \rangle,\!\langle\!k, j \rangle}^{(t)} 
        (\!\mathbf{f}_l^\mathrm{H}\!\mathbf{H}_{\!\langle\!k, i \rangle}^{(t)}\!)^\mathrm{H} 
        \mathbf{f}_l^\mathrm{H}\!\mathbf{H}_{\!\langle\!k, j \rangle}^{(\!t\!)}\!\mathbf{u}_{\langle\!k, j \rangle}.\!
    \end{align}
\end{subequations}
Note that $\mathbf{A}_{\langle k, i \rangle}^{(t)}$ is a positive semidefinite matrix, and that the constraint \eqref{constraint_u} is convex w.r.t. $\mathbf{u}_{\langle k, i \rangle}$. Thus, P2 is a convex quadratically constrained quadratic programming (QCQP) problem w.r.t. $\mathbf{u}_{\langle k, i \rangle}$, which can be solved by standard convex optimization tools.

\subsubsection{Optimizing \texorpdfstring{$\mathbf{f}_k$}{Lg} with fixed \texorpdfstring{$\mathbf{u}$}{Lg}}
\label{opt_f}
Similarly to \ref{opt_u}, with fixed $\mathbf{u}$, \eqref{P3obj} reduces to 
\begin{mini*}|s|
    {\substack{\mathbf{f}_k}}
    {\mathbf{f}_k^\mathrm{H} \mathbf{A}_k^{(t)} \mathbf{f}_k - 2\Re \{ \mathbf{b}_k^{(t)}{}^\mathrm{H} \mathbf{f}_k\}}
    {}
    {\text{(P3)}:}
\end{mini*}
where $\mathbf{A}_k^{(t)} \in \mathbb{C}^{N_{\mathrm{R}} \times N_{\mathrm{R}}}$ and $\mathbf{b}_k^{(t)} \in \mathbb{C}^{N_{\mathrm{R}}}$ are denoted by
\begin{subequations}
    \label{Ab_PS}
    \begin{align}
        \mathbf{A}_k^{(t)} 
        &\triangleq y_k^2 \sum\nolimits_{l=1}^K \sum\nolimits_{i,j \in \mathcal{M}_l} \rho_{\langle l, i \rangle, \langle l, j \rangle}^{(t)}  \mathbf{H}_{\langle l, i \rangle}^{(t)} \mathbf{u}_{\langle l, i \rangle} (\mathbf{H}_{\langle l, j \rangle}^{(t)} \mathbf{u}_{\langle l, j \rangle})^\mathrm{H} + (y_k^2 \sigma^2/2) \mathbf{I}_{N_{\mathrm{R}}}\\
        \mathbf{b}_k^{(t)} 
        &\triangleq \frac{y_k }{ \sum_{i \in \mathcal{M}_k} Q_{\langle k, i \rangle} }  \sum_{i,j \in \mathcal{M}_k} \rho_{\langle k, i \rangle, \langle k, j \rangle}^{(t)} Q_{\langle k, j \rangle} \mathbf{H}_{\langle k, i \rangle}^{(t)} \mathbf{u}_{\langle k, i \rangle}.
    \end{align}
\end{subequations}
P3 is convex w.r.t. $\mathbf{f}_k$ by noting the positive semidefinite matrix $\mathbf{A}_k^{(t)}$. \revise{ Note that $\mathbf{f}_k$ here is obtained by dropping the unit-length constraint in \eqref{constraint_f}. Thus, we finally obtain the optimal $\mathbf{f}_k$ by scaling the obtained $\mathbf{f}_k$ to a unit vector.} We summarize the optimization of $\mathbf{f}$ and $\mathbf{u}$ in Algorithm \ref{alg:beamforming}.
\begin{algorithm}[htb]
\caption{ AO Algorithm to Optimize $\mathbf{f}$ and $\mathbf{u}$} 
\label{alg:beamforming} 
\begin{algorithmic}[1] 
\REQUIRE $\mathcal{M}, \{\boldsymbol{\rho}_k^{(t)}, \mathbf{H}_{k,\langle l, i \rangle}^{(t)}, Q_{\langle k, i \rangle} | k,l \in [K], i \in [M_k]\},$ and $I_{\max}$. 
\STATE {\textbf{Initialization:}} $\mathbf{y}$, $\mathbf{f}$ and $\mathbf{u}$.
\FOR{ $\tau \in [I_{\max}]$ }
    \FOR{$k \in [K]$ }
        \FOR{$i \in \mathcal{M}_k$ }
            \STATE Compute $\mathbf{A}_{\langle k, i \rangle}^{(t)}, \mathbf{b}_{\langle k, i \rangle}^{(t)}$ based on \eqref{Ab_device} ; 
            \STATE Optimize $\mathbf{u}_{\langle k, i \rangle}$ by solving (P2);
        \ENDFOR
        \STATE Compute $\mathbf{A}_k^{(t)}, \mathbf{b}_k^{(t)}$ based on \eqref{Ab_device};
        \STATE Optimize $\mathbf{f}_k$ by solving (P3), update $\mathbf{f}_k$ (by scaling to a unit vector) 
        \STATE Updates $\mathbf{y}$ based on \eqref{y_define};
    \ENDFOR
\ENDFOR
\ENSURE $(\mathbf{f}, \mathbf{u})$.
\end{algorithmic}
\end{algorithm}

\subsection{Optimizing \texorpdfstring{$\mathcal{M}$}{Lg} with fixed \texorpdfstring{$\mathbf{f}$}{Lg} and \texorpdfstring{$\mathbf{u}$}{Lg} } 
In this subsection, we optimize device selection set $\mathcal{M}$ with fixed beamforming $\mathbf{f}$ and $\mathbf{u}$. For notational convenience, we use a binary indication vector $\mathbf{s}$ to represent the device selection set $\mathcal{M}$, i.e., $\mathbf{s} \triangleq [\mathbf{s}_{1}^\mathrm{T}, \cdots, \mathbf{s}_k^\mathrm{T}]^\mathrm{T} \in \{0,1\}^M$, where $\mathbf{s}_k \in \{0,1\}^{M_k}$ with the $i$-th entry $s_{\langle k, i \rangle} = 1$ meaning that the $\langle k, i \rangle$-th device is selected and $s_{\langle k, i \rangle} = 0$ otherwise. Note that $\mathcal{M}$ and $\mathbf{s}$ can be interchanged.
Since $\mathbf{s}$ is discrete, we introduce the Gibbs sampling to optimize $\mathbf{s}$ by following the approach in \cite{liu2020reconfigurable}. 
We denote $\mathbf{s}^{\text{old}}$ as the sampling device selection solution obtained from the proceeding sampling round. At the current sampling round, the sampling set $\mathcal{S}$ is generated from $\mathbf{s}^{\text{old}}$, given by $\mathcal{S} \triangleq \{\mathbf{s}^{\text{old}}\} \cup \{ \mathbf{s}^{\text{old}}_{\langle k, i \rangle} | k \in [K], i \in [M_k]\}$, where $\mathbf{s}^{\text{old}}_{\langle k, i \rangle}$ denotes the indication vector that differs from $\mathbf{s}^{\text{old}}$ only at the $\langle k, i \rangle$-th element, corresponding to the $\langle k, i \rangle$-th device. We sample $\mathbf{s}^{\text{new}}$ according to the following distribution $\pi(\mathbf{s}^{\text{new}})$:
\begin{equation}
    \label{distribution}
    \pi(\mathbf{s}^{\text{new}}) \triangleq 
    \frac{\exp{\left(-\phi(\mathbf{s}^{\text{new}})/\beta \right) }}{\sum_{k=1}^{K} \sum_{i=1}^{M_k} \exp{\left(-\phi(\mathbf{s}^{\text{old}}_{\langle k, i \rangle}) /\beta \right)} + \exp{\left(-\phi(\mathbf{s}^{\text{old}}) /\beta \right)}},
\end{equation}
where $\phi(\mathbf{s}^{\text{old}}_{\langle k, i \rangle})$ denotes the objective $\mathcal{E}^{(t)}(\mathcal{M}, \mathbf{f}, \mathbf{u})$ with $\mathbf{f}$ and $\mathbf{u}$ obtained by Algorithm \ref{alg:beamforming} with the device selection set $\mathcal{M}$ corresponding to $\mathbf{s}^{\text{old}}_{\langle k, i \rangle}$, and $\beta > 0$ denotes the ``temperature parameter" to accelerate convergence. We summarize the overall algorithm for the optimization of $\mathcal{M}$, $\mathbf{f}$ and $\mathbf{u}$ in Algorithm \ref{alg:selection}. 
\begin{algorithm}[htb]
\caption{AO Algorithm plus Gibbs Sampling} 
\label{alg:selection} 
\begin{algorithmic}[1] 
\REQUIRE $\{\boldsymbol{\rho}_k^{(t)}, \mathbf{H}_{\langle l, i \rangle}^{(t)}, Q_{\langle k, i \rangle} | k,l \in [K], i \in [M_k]\}, J_{max}, \beta, \gamma$. 
\STATE {\textbf{Initialization:}} $\mathbf{s}^{\text{old}} = \mathbf{1}_{M \times 1}$.
\FOR{ $j \in [J_{max}]$ }
    \STATE $\mathbf{s}^{\text{old}} = \mathbf{s}^{\text{new}}$;
    \STATE Generate $\mathcal{S}$;
        \FOR{every $\mathbf{s}^{\text{old}}_{\langle k, i \rangle} \in \mathcal{S}$}
            \STATE Optimize $(\mathbf{f}, \mathbf{u})$ by solving P2 with given $\mathbf{s}^{\text{old}}_{\langle k, i \rangle}$, with Algorithm \ref{alg:beamforming};
        \ENDFOR
    \STATE Sample $\mathbf{s}^{\text{new}}$ according to \eqref{distribution};
    \STATE Refresh $\beta = \gamma \beta$ for a certain $\gamma \in (0,1)$;
\ENDFOR
\ENSURE $\mathbf{s}^{\text{new}}$ with corresponding $(\mathbf{f}, \mathbf{u})$.
\end{algorithmic}
\end{algorithm}

\subsection{Complexity Analysis}
We now briefly discuss the computational complexity involved in Algorithms~\ref{alg:beamforming} and \ref{alg:selection}. For Algorithm~\ref{alg:beamforming}, both P2 and P3 are QCQP problems that can be solved by existing optimization solvers based on the interior point method. Thus, the worst-case complexity of Algorithm \ref{alg:beamforming} is given by $\mathcal{O}(I_{\max} (K+M) N^{3.5})$, where $N \triangleq \max \{N_{\mathrm{T}}, N_{\mathrm{R}}\}$ denotes the maximum number of the transmit or receive antennes, $I_{\max}$ is the max iteration times of optimization, and $M$ is the total number of the devices in the OA-FMTL. Algorithm \ref{alg:selection} invokes Algorithm \ref{alg:beamforming} for $J_{max}M$ times to optimize device selection set $\mathcal{M}$. Thus, the complexity of Algorithm \ref{alg:selection} is $\mathcal{O}(J_{max} I_{\max} (KM+M^2) N^{3.5})$. 

We note that the complexity of Algorithm \ref{alg:selection} is quadratic in $M$, due to the use of Gibbs sampling in the optimization of device selection. When the number of devices $M$ in the OA-FMTL is large, Gibbs sampling causes a tremendous computation burden. Recall from Section \ref{sec:zeta_select} that the device selection is adopted by the existing works to reduce the impact of stragglers. As pointed out in Section \ref{sec:zeta_select}, our proposed scheme optimises the weighting factor $\zeta_k$ to minimize the communication MSE, which relieves the straggler problem significantly. We observe from experiments that the improvement of device selection in Algorithm \ref{alg:selection} is negligible, as compared to the performance of Algorithm \ref{alg:beamforming}. Therefore, we prefer to use Algorithm \ref{alg:beamforming} in the system optimization, since Algorithm \ref{alg:beamforming} has a much lower complexity.

\section{Numerical Results}
\subsection{Simulation Setups}
We consider a three-dimensional (3-D) simulation scenario as shown in Fig.~\ref{fig:BS}. 
The point locations are represented by cylindrical coordinate triples $(\delta, \theta, \chi)$, where $\delta$, $\theta$ and $\chi$ denote the radial distance, the azimuth, and the height, respectively. The locations of the devices are distributed as follows.
All the devices in the OA-FMTL are located in a circle with center $O = (0,0,0)$ and radius $\Delta$. We set the location of the $\langle k, i \rangle$-th device to $(\delta_{\langle k, i \rangle}, \theta_{\langle k, i \rangle}, 0)$, where $\delta^2_{\langle k, i \rangle}$ is uniform in $[0,\Delta]$, and $\theta_{\langle k, i \rangle}$ is uniform in $[0,2\pi)$.
\revise{ The PS is placed at the center of the circle, i.e.,$(0, 0, 10)$. 
We adopt the channel model in \cite{goldsmith2005wireless}, given by
$\mathbf{H}_{\langle l,i \rangle} = \sqrt{G_{\mathrm{S}} G_{\mathrm{D}}\kappa \tilde{\delta}_{\langle l,i \rangle}^{-\alpha}} \mathbf{\tilde{H}}_{\langle l,i \rangle}$,
where the entries of $\mathbf{\tilde{H}}_{\langle l,i \rangle}$ are modeled as i.i.d. circularly symmetric complex Gaussian (CSCG) random variables with zero-mean and unit-variance, $G_{\mathrm{S}}$ and $G_{\mathrm{D}}$ are the antenna gains at the PS and the devices, respectively, $\kappa$ is the path loss at the reference distance $\delta_0 = \SI{1}{m}$ \cite{wu2019intelligent}, $\alpha$ is the path loss exponent, and $\tilde{\delta}_{\langle l,i \rangle} \triangleq \sqrt{\delta_{\langle l, i \rangle}^2 + 10^2}$ is the distance between the $\langle l,i \rangle$-th device and the PS.}
The simulation settings are given in Table \ref{SimuPara}.
\begin{figure}[htbp]
    \centering
    \subfloat[][]{
        \includegraphics[width=0.35\linewidth]{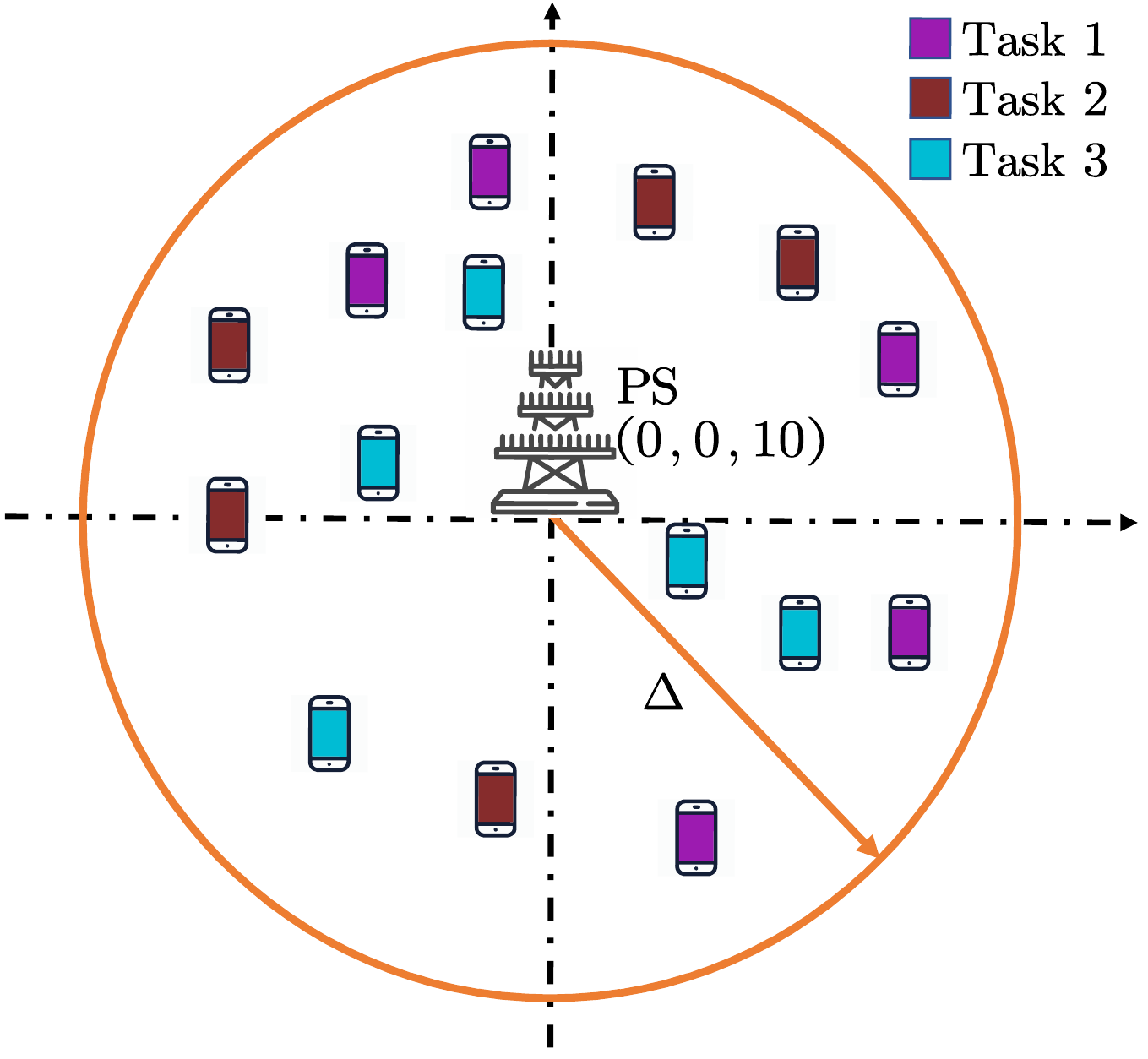}
    }
    \caption{ \revise{ An illustration of locations for the devices and the PS in the OA-FMTL on the vertical view.} }
    \label{fig:BS}
\end{figure}

\begin{table}[ht]
\caption{System Parameters}
\small 
\centering
\begin{tabular}{|c|c||c|c||c|c||c|c|}
\hline    
Parameter&Value &Parameter&Value &Parameter&Value &Parameter&Value\\
\hhline{|--||--||--||--|}    
$N_{\mathrm{T}}$&$2$       &$N_{\mathrm{R}}$&$8$        &$I_{\max}$&$50$      &$J_{max}$&$50$     \\ 
\hline    
$\alpha$&$3.8$         &$\kappa$&$\SI{-60}{dB}$  &$G_{\mathrm{S}}$&$\SI{5}{dBi}$ &$G_{\mathrm{D}}$&$\SI{0}{dBi}$\\ 
\hline    
$P_0$&$\SI{1}{W}$     &$\beta$&$1$  &$\gamma$&$0.9$   &$\Delta$&$\SI{100}{m}$\\ 
\hline 
$\sigma^2$&$\SI{-80}{dBm}$     &$Q_k$ &$60000$  & & & &\\ 
\hline
\end{tabular}
\label{SimuPara}
\end{table}

We set three image classification tasks as FL tasks in the OA-FMTL, with each FL task being trained on an individual dataset, i.e., MNIST for task 1, Fashion-MNIST for task 2 and KMNIST for task 3. 
\revise{ For each FL task, we train a CNN with two $5\!\times\!5$ convolution layers (the first with $16$ channels, the second with $32$, each followed by $2\!\times\!2$ max pooling), a fully connected layer with $50$ units and ReLu activation, and a final softmax output layer ($D\!=\!39408$ total parameters).}
The loss function is the cross-entropy loss. 
\revise{
We study two ways of partitioning the dataset $\mathcal{A}_k$ over devices: 1) \textbf{i.i.d.}, where the data are shuffled, and then assigned evenly to the $M_k$ devices; 2) \textbf{Non-i.i.d.}, where each device randomly selects $5$ classes, and then randomly draws $\frac{Q_k}{5 M_k}$ samples from each selected class.
}

\subsection{Comparisons of the Proposed Algorithms Under Various Settings}
In this subsection, we study the impact of various approximations of the correlation matrices $\{\boldsymbol{\rho}_k^{(t)}\}$. Specifically, we consider the following three approximations of the correlation matrix:
\begin{itemize}
\item \textbf{Approximation 1}: 
$\boldsymbol{\rho}_k^{(t)} = \frac{1}{D} \sum_{d=1}^D \mathbf{z}_{\langle k,d \rangle}^{(t)} \mathbf{z}_{\langle k,d \rangle}^{(t)}{}^\mathrm{T}$.
\revise{
\item \textbf{Approximation 2}: $\boldsymbol{\rho}_k^{(t)}$ is approximated by $\boldsymbol{\rho}_k^{(t)} = \epsilon \mathbf{1}_{M_k} + (1- \epsilon) \mathbf{I}_{M_k}$, where $\epsilon$ is an empirical parameter to represent the degree of correlation.
}
\end{itemize}
\revise{ Approximation 1 estimates the spatial correlation between the gradients of the devices from the same task $k$ at each round. However, Approximation 1 is impractical since the calculation of $\boldsymbol{\rho}_k^{(t)}$ requires the knowledge of the local gradients uploaded by the devices. In contrast, Approximation 2 is more practical in implementation, where $\epsilon$ can be set empirically. Therefore, the proposed schemes with Approximation 1 are only used as a performance baseline, and the schemes with Approximation 2 are used in performance comparison with other counterpart schemes.}
In addition, we adopt the error-free case as a performance upper bound of each FL task in the OA-FMTL:
\begin{itemize}
\item Error-free bound: Each FL task is trained independently with all the devices being selected, and the model aggregation is error-free at each communication round. 
\end{itemize}
\begin{figure}[htbp]
    \centering
    \includegraphics[width=0.7\linewidth]{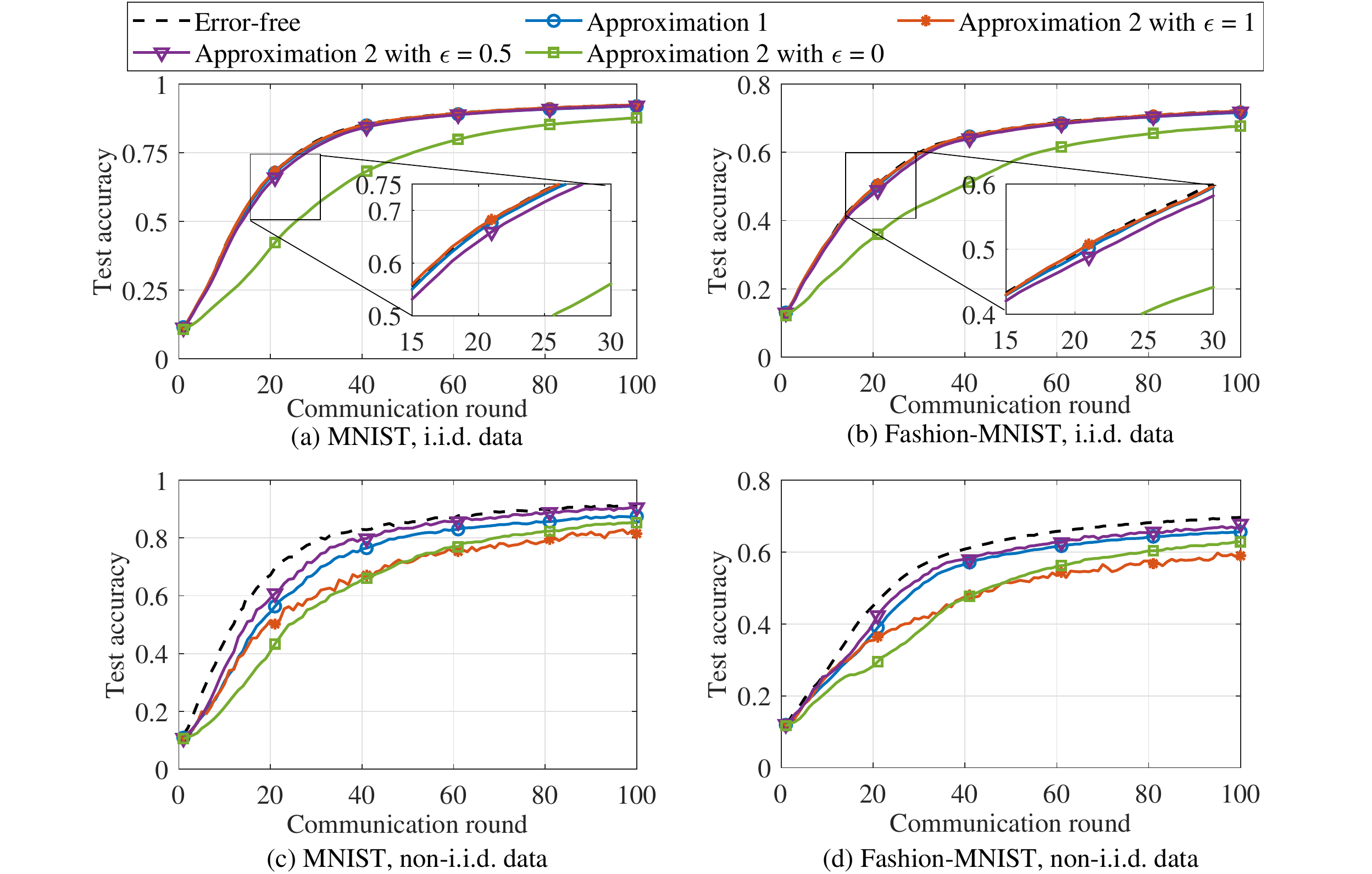}
    \caption{ \revise{ FL test accuracy of the proposed Algorithm \ref{alg:beamforming} with different approximations of $\boldsymbol{\rho}_k^{(t)}$ versus the communication rounds, with $K = 2$, $M_1 = M_2 = 20$. Left: MNIST; right: Fashion-MNIST; top: i.i.d. data; bottom: non-i.i.d. data.} }
    \label{fig:corr_acc}
\end{figure}

We plot the FL test accuracy curve of the proposed Algorithm \ref{alg:beamforming} with the above approximations of the correlation matrix $\boldsymbol{\rho}_k^{(t)}$ in Fig.~\ref{fig:corr_acc}. 
\revise{ We train two FL tasks, namely, tasks 1 and 2, both on i.i.d. data and on non-i.i.d. data. We set the numbers of devices $M_1\!=\!M_2\!=\!20$, the learning rates $\eta_1\!=\!\eta_2\!=\!0.05$, and the momentum $ = 0.5$. The local updates consist of $5$ times of SGD. The results are averaged over $20$ Monte Carlo trials. In Approximation 2, we set $\epsilon\!=\!0$, $\epsilon\!=\!0.5$, and $\epsilon\!=\!1$.
From Fig.~\ref{fig:corr_acc}, we see that in the case of i.i.d. data, both Approximation 1 and Approximation 2 with $\epsilon\!=\!1$ achieve test accuracies close to the error-free bound on both two tasks, and Approximation 2 with $\epsilon\!=\!0$ has the worst learning performance. This is because Approximation 2 with $\epsilon\!=\!0$ suffers from a serious aggregation error for ignoring the correlation between the local gradients. 
On the other hand, in the case of non-i.i.d. data, the accuracies achieved by Approximation 1 and Approximation 2 with $\epsilon\!=\!0.5$ are close to the error-free bound, since $\epsilon\!=\!0.5$ approximates the spatial correlation more precisely in the case of non-i.i.d. data. 
Thus, Approximation 2 with $\epsilon\!=\!1$ for i.i.d. data and Approximation 2 with $\epsilon\!=\!0.5$ for non-i.i.d. data are preferred in the system design. 

We next study the necessity of device selection by comparing the proposed Algorithms \ref{alg:beamforming} and \ref{alg:selection}. We simulate Algorithms \ref{alg:beamforming} and \ref{alg:selection} on tasks $1$ and $2$, with $\epsilon\!=\!1$ on i.i.d. data and $\epsilon\!=\!0.5$ on non-i.i.d. data. 
The numbers of devices are set to $M_1 = M_2 = 10$. The learning rates are set to $\eta_1 = \eta_2 = 0.05$, and the momentum is set to $0.5$. The local updates have $10$ times of SGD. The results are averaged over $10$ Monte Carlo trials
\footnote{\revise{
    Note that we choose a relatively small number of trials due to the high computational complexity of device selection in Algorithm \ref{alg:selection}. In simulations, we use a personal computer with an Intel(R) Core(TM) i7-10700 CPU and a GTX 1050Ti GPU. One Monte Carlo trial of Algorithm 3 takes about $10$ hours.}
}. 
In Fig.~\ref{fig:samp_acc}, we present the test accuracy of Algorithms \ref{alg:beamforming} and \ref{alg:selection} versus communication rounds on i.i.d. and non-i.i.d. data. From Fig.~\ref{fig:samp_acc}, we see that Algorithms \ref{alg:beamforming} and \ref{alg:selection} always perform closely.
}
This implies that device selection is no longer necessary to our proposed scheme, which avoids the high computational complexity involved in the implementation of device selection. Thus, we henceforce always employ Algorithm \ref{alg:beamforming} performance comparison when we refer to the proposed AO algorithm.

\begin{figure}[htbp]
    \centering
    \includegraphics[width=0.7\linewidth]{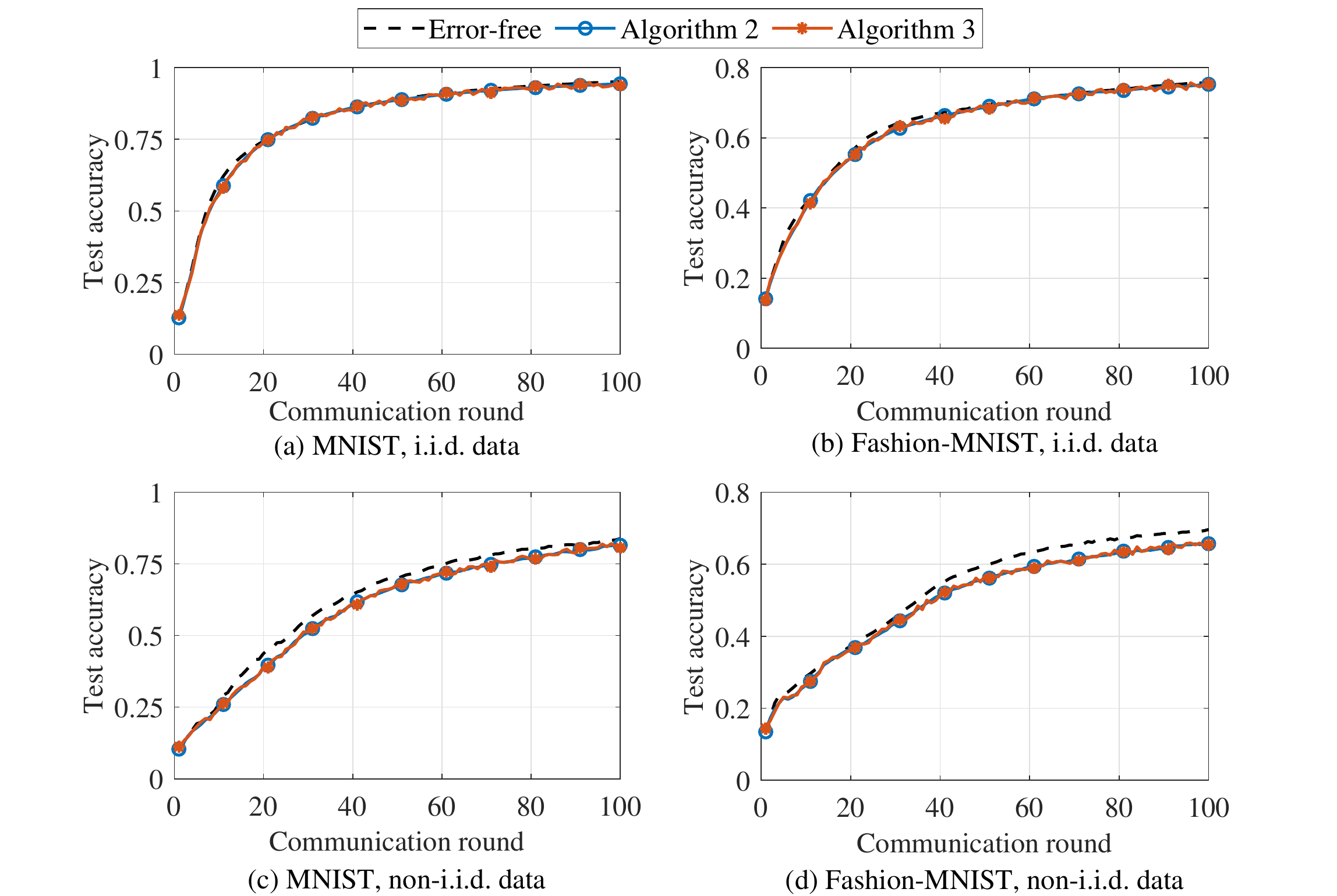}
    \caption{ \revise{ Test accuracy versus communication rounds, with $K = 2$, $M_1 = M_2 = 10$. Left: MNIST; right: Fashion-MNIST; top: i.i.d. data; bottom: non-i.i.d. data.} }
    \label{fig:samp_acc}
\end{figure}

\subsection{Comparisons With Existing Schemes}
In this subsection, we present the performance of system optimization obtained by the proposed AO algorithm on i.i.d. data. We consider the following baselines for comparison: 
\begin{itemize}
\item SOCP-based cooperative power control \cite{cao2021cooperative}: This method assumes that all the devices selected, and that the local gradients from devices in the same task are independent with each other. The phases of transmit beamforming $\mathbf{u}_{\langle k,i \rangle}, \forall k, i$ are given by the zero-forcing. The design of transmit power $\|\mathbf{u}_{\langle k,i \rangle}\|^2, \forall k, i$ are formulated as an second-order cone programming (SOCP) problem solved by the bisection method. 
\item SCA-based optimization and device selection \cite{liu2020reconfigurable}: 
\revise{ $\mathcal{M}_k$, $\mathbf{f}_k$, and $\{\mathbf{u}_{\langle k, i \rangle}\}$ for each FL task are optimized separately.} For task $k$, the receive beamforming $\mathbf{f}_k$ is optimized by the successive convex approximation (SCA)-based optimization algorithm, with given transmit beamforming $\mathbf{u}_k$ and the weighting factor $\zeta_k$ determined by zero-forcing. With given optimized $\mathbf{f}_k$, $\mathbf{u}_k$ and $\zeta_k$, device selection set $\mathcal{M}_k$ is optimized via Gibbs sampling.
\item Receive beamforming by differential geometry programming \cite{zhu2018mimo}: This method optimizes each task separately, with all the devices selected. $\mathbf{f}_k$ is optimized on the Grassmann manifold via differential geometry programming, and $\mathbf{u}_k$ is given by the zero-forcing.
\item Difference-of-convex (DC) programming and device selection \cite{yang2020federated}: This method optimizes each task separately. For each task, the method maximizes the number of selected devices by a two-step framework based on DC programming, with a given threshold of the communication MSE. 
\end{itemize}

\revise{ Here we simulate the case of $3$ FL tasks on i.i.d. data. The numbers of devices for the three tasks are set to $M_1\!=\!M_2\!=\!M_3\!=\!20$. The learning rates are set to $\eta_1\!=\!\eta_2\!=\!\eta_3\!=\!0.05$. The local updates contain $5$ mini-batches of SGD. The noise power is $\sigma^2 = \SI{-60}{dBm}$. We set $\epsilon\!=\!0.5$.}
\revise{
Besides, we introduce the normalized mean square error (NMSE) of each task $k$ at round $t$, defined as
$\mathrm{NMSE}^{(t)}_k\!\triangleq 10\log_{10}\!\Big( \mathbb{E} \big [ \|\hat{\mathbf{g}}_{k}^{(t)} - \mathbf{g}_{k}^{(t)}\|_2^2 / \|\mathbf{g}_{k}^{(t)}\|_2^2 \big]\Big).
$
In Table~\ref{tab:NMSE}, we list the average NMSE over $20$ Monte Carlo trials at $t=40$ and $t=90$. Benefiting from interference awareness, the proposed AO algorithm and the method in \cite{cao2021cooperative} achieve much better NMSEs on all the three FL tasks than the other methods.} We also see that our algorithm significantly outperforms the method in \cite{cao2021cooperative}. This is attributed to the careful optimization of the transceiver beamforming based on the proposed analytical framework.
\begin{table}[htbp]
    \caption{Communication NMSE}
    \small
    \centering
    \begin{tabular}{l|cccccc}
    \hline
    \multicolumn{1}{c|}{\multirow{3}{*}{Optimization method}}              & \multicolumn{6}{c}{NMSE ($\SI{}{dB}$)}                                            \\ \cline{2-7} 
    \multicolumn{1}{c|}{}                                                  & \multicolumn{2}{c|}{MNIST ($k = 1$)}    & \multicolumn{2}{c|}{Fashion-MNIST ($k = 2$)} & \multicolumn{2}{c}{KMNIST ($k = 3$)}\\ \cline{2-7} 
    \multicolumn{1}{c|}{}   & $t=40$ & \multicolumn{1}{c|}{$t=90$} 
    & $t=40$     & \multicolumn{1}{c|}{$t=90$}       & $t=40$ & $t=90$     \\ 
    \hline
    Proposed AO algorithm  & $\mathbf{-1.43}$  & \multicolumn{1}{c|}{$\mathbf{-1.37}$} 
        & $\mathbf{-1.81}$     & \multicolumn{1}{c|}{$\mathbf{-1.40}$} 
        & $\mathbf{-1.98}$     & $\mathbf{-1.75}$    \\
    SOCP-based cooperative power control \cite{cao2021cooperative} & $-0.84$  & \multicolumn{1}{c|}{$-1.01$} 
        & $-1.10$        & \multicolumn{1}{c|}{$-1.01$}   & $-1.11$     & $-0.98$    \\
    SCA \& Gibbs \cite{liu2020reconfigurable}                & $-0.40$      & \multicolumn{1}{c|}{$-0.47$}  
        & $-0.50$      & \multicolumn{1}{c|}{$-0.52$}  & $-0.58$    & $-0.52$     \\
    Differential geometry \cite{zhu2018mimo}                 & $-0.16$    & \multicolumn{1}{c|}{$-0.20$} 
        & $-0.21$      & \multicolumn{1}{c|}{$-0.19$}  & $-0.22$    & $-0.21$    \\
    DC and device selection \cite{yang2020federated}         & $-0.17$    & \multicolumn{1}{c|}{$-0.20$} 
        & $-0.15$      & \multicolumn{1}{c|}{$-0.21$}  & $-0.06$    & $-0.13$    \\ 
    \hline
    \end{tabular}
    \label{tab:NMSE}
\end{table}

\begin{figure}[htbp]
    \centering
    \includegraphics[width=0.9\linewidth]{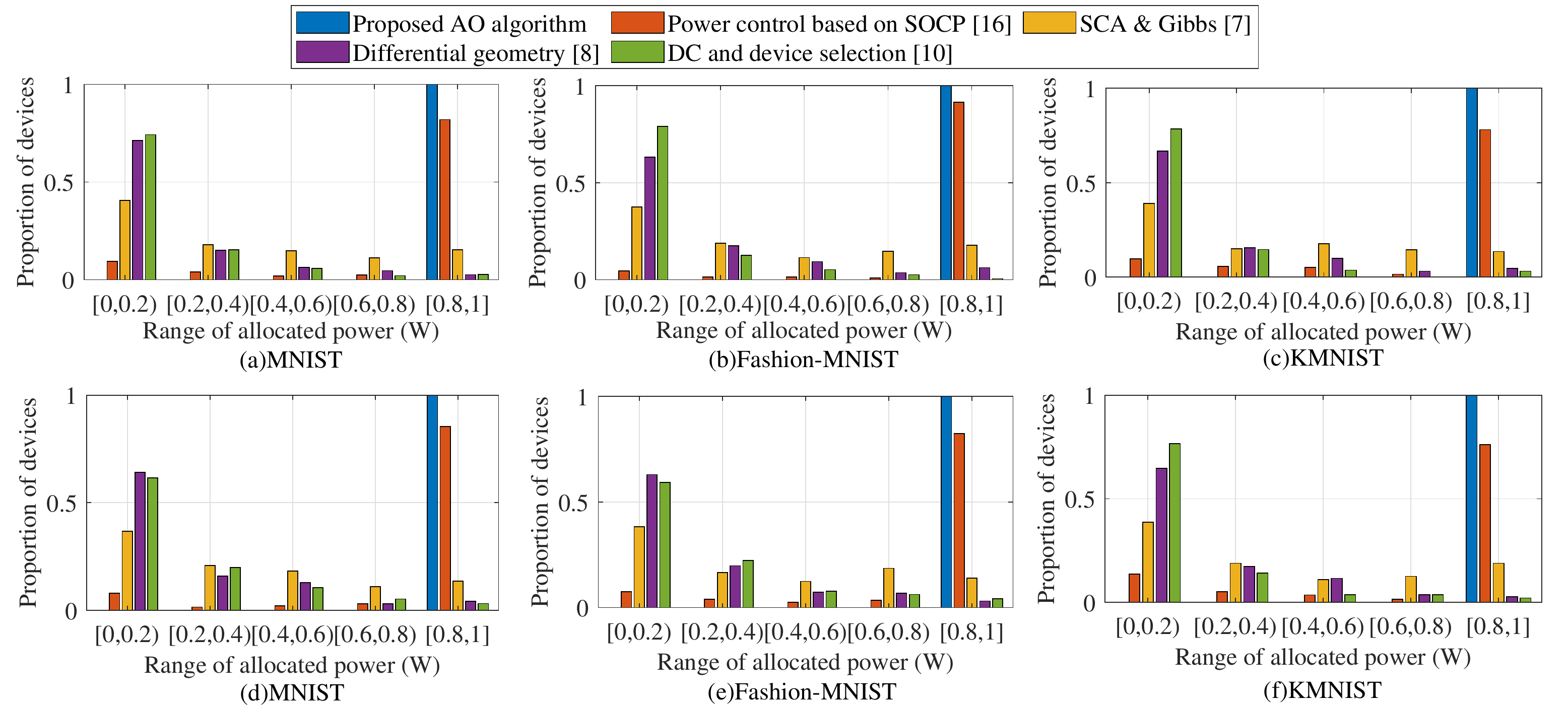}
    \caption{Proportion of devices versus range of allocated power with tasks 1, 2 and 3. $t = 40$ for (a)-(c); $t = 90$ for (d)-(f).}
    \label{fig:opt_pow}
\end{figure}
In Fig.~\ref{fig:opt_pow}, we plot the histogram of allocated transmission powers for various optimization methods. The methods based on zero-forcing \cite{cao2021cooperative, zhu2018mimo, yang2020federated} only allocate full power to less than $20\%$ of all the devices, due to the stragglers with the worst channel conditions. The SCA \& Gibbs algorithm in \cite{liu2020reconfigurable} excludes several stragglers through Gibbs sampling, improving the number of full-power-allocated devices, but the percentage is still below $50\%$. We see that the transmission powers of most devices are allocated fully in the proposed AO algorithm. This is because our proposed scheme relaxes the hard requirement for all the devices to align their gradients with the stragglers, which gives freedom to the devices to fully exploit the power budgets.

In Fig.~\ref{fig:opt_acc}, we present the test accuracies of various optimization algorithms versus communication rounds. As shown in Fig.~\ref{fig:opt_acc}, the proposed algorithm achieves an accuracy close to the error-free bound in all the three FL tasks and significantly outperforms the other baselines, which clearly demonstrates the superiority of our proposed scheme. 
\begin{figure}[ht]
    \centering
    \includegraphics[width=0.9\linewidth]{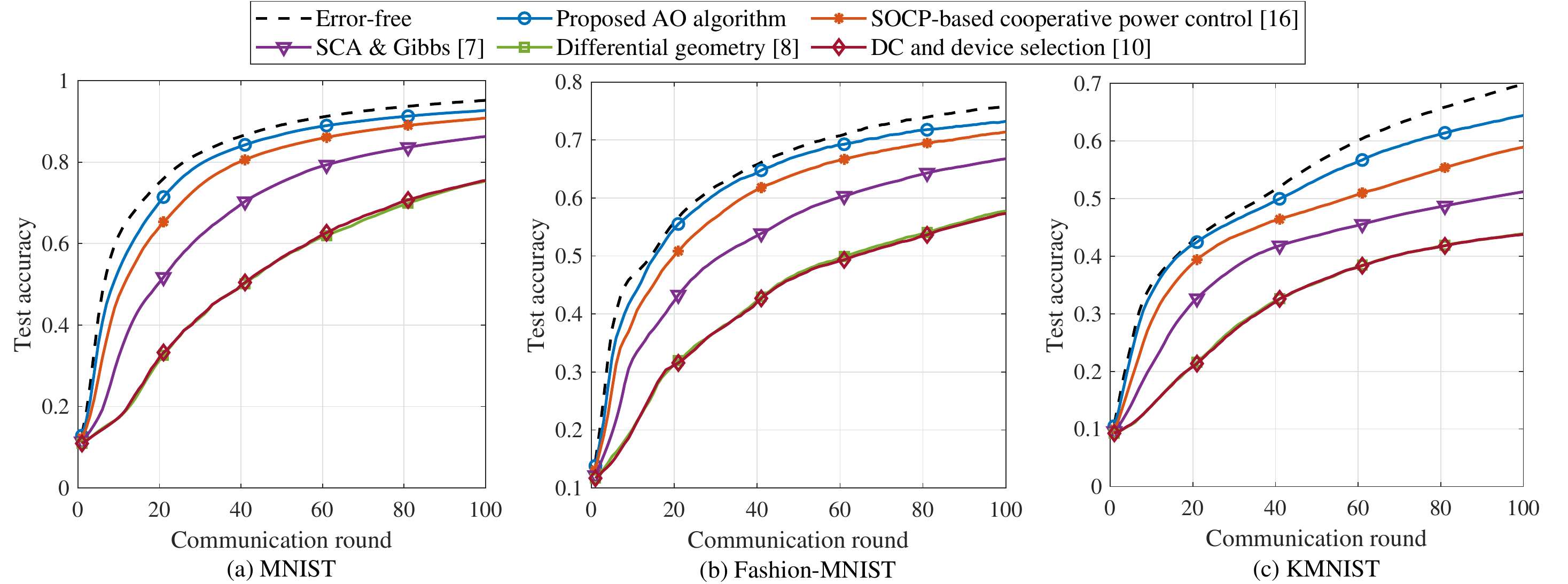}
    \caption{Test accuracy versus communication rounds on (a) MNIST, (b) Fashion-MNIST, and (c) KMNIST, with $K = 3$, $M_1 = M_2 = M_3 = 20$, i.i.d. data.}
    \label{fig:opt_acc}
\end{figure}

\section{Conclusion}
In this paper, we studied a problem of designing an OA-FMTL system over MIMO MAC channel. 
We proposed a misalignment-tolerant strategy to align the local gradients at model aggregation at the PS side to relieve the straggler problem. We further derived a communication-learning framework to analyze the OA-FMTL performance by characterizing the performance loss due to device selection, inter-task interference and communication noise. Based on the analytical framework, we formulated an optimization problem with respect to device selection, transmit beamforming, and receive beamforming. We captured the spatial correlation between the local gradients to enhance the optimization and proposed a low-complexity algorithm to solve the communication-learning problem based on AO framework. Finally, we performed extensive numerical experiments to demonstrate the learning accuracy outstanding improvements of the proposed algorithm by comparison with the state-of-the-art methods.

\appendices
\section{Proof of Theorem~\ref{Theorem:MSE}}
\label{Appendix:Prf_Theorem:MSE}

From [\citenum{friedlander2012hybrid}, 3.2], the device selection MSE $\mathbb{E}[ \|\mathbf{e}_{\mathrm{ds},k}^{(t)} \|^{2}]$ is bounded by \eqref{ErrorDevicesSelection}.
We next consider the communication MSE of the $k$-th task given by
\begin{equation}
    \label{MSE}
    \begin{aligned}
        \mathbb{E}\left[\|\mathbf{e}_{\mathrm{com},k}^{(t)}\|^{2}\right] 
        = & \sum\nolimits_{d=1}^{D} \mathbb{E}\left[\left|g_k^{(t)}[d] - \hat{g}_k^{(t)}[d] \right|^2\right].
    \end{aligned}
\end{equation}
By plugging \eqref{idealgradient}, \eqref{nomalize}, \eqref{modulation} and \eqref{receivedgradient} into \eqref{MSE}, we obtain 
\begin{align}
    \label{MSEplug}
    \mathbb{E}[\|\!\mathbf{e}_{\!\mathrm{com},k}^{(\!t\!)}\!\|^{2}] 
    \!\!=\!\!\frac{1}{\big(\!\!\sum\limits_{i \in \mathcal{M}_k}\!\!Q_{\langle k, i \rangle}\!\big)^{2}}\!\!\sum_{c=1}^{C}\!
    \mathbb{E}\!\Bigg[\!\Big|\!\!\sum_{i \in\!\mathcal{M}_{\!k}}\!Q_{\!\langle \!k, i \rangle}\!\!\sqrt{\!v_k^{\!(\!t\!)}}\!r_{\!\langle\!k, i \rangle}^{(\!t\!)}\![c]\!- 
    \!\!\!\sum_{l\!=\!1}^K\!\!\sum_{i \in\!\mathcal{M}_{\!l}}\!\zeta_k\!\mathbf{f}_{\!k}^\mathrm{H}\!\mathbf{H}_{\langle\!l, i \rangle}^{(\!t\!)}\!\mathbf{u}_{\!\langle\!l, i \rangle}\!r_{\!\langle\!l, i \rangle}^{(\!t\!)}\![c]\!
    \!-\!\zeta_{\!k}\!\mathbf{f}_{\!k}^\mathrm{H}\!\mathbf{n}^{(\!t\!)}\![c]\!\Big|^2\!\Bigg]\!.
\end{align}
Based on Assumption~\ref{asp:g_task}, 
for $\forall k \neq l$, we have $\mathbb{E}[r_{\langle k, i \rangle}[c] r_{\langle l, j \rangle}[c]^\dagger] = 0$.
Thus, we obtain \eqref{ErrorCommunication} by expanding \eqref{MSEplug}. 
What remains is to prove \eqref{F_Upperbound}. From [\citenum{friedlander2012hybrid}, Lemma~2.1], Assumptions~\ref{asp:f_Lipschitz}-\ref{asp:f_bound} lead to an upper bound of $F_k(\cdot)$ at the $t$-th round. Thus, we have the following lemma.
\begin{lemma}
    Assume that $F_k(\cdot)$ satisfies Assumptions~\ref{asp:f_Lipschitz}-\ref{asp:f_bound}, at the $t$-th communication round with the learning rate $\eta_k$ is set to $1/\omega_k$. Then
    \begin{equation}
        \mathbb{E}[F_k(\mathbf{w}_k^{(t+1)})]\!\leq\! \mathbb{E}[F_k(\mathbf{w}_k^{(t)})]\!-\!\frac{1}{2 \omega_k} \mathbb{E}[\|\nabla F_k(\mathbf{w}_k^{(t)})\|^{2}] + \!\frac{1}{2\omega_k} \mathbb{E}[\|\mathbf{e}_k^{(t)}\|^{2}],
        \label{LemmaEq:upper bound}
    \end{equation}
    where $\omega_k$ is the Lipschitz continuity parameter defined in \eqref{Lipschitz}, and $\mathbb{E}[\cdot]$ is the expectation w.r.t. $\{ n_{k,z}^{(t)}[c], g_{\langle k, i \rangle}^{(t)}[d] | k \in [K], z \in [N_{\mathrm{R}}], c \in [C], i \in [M_k], d \in [D], \tau \in [t+1] \} $.
    \label{Lemma:Ef}
\end{lemma}
\begin{proof}
    See [\citenum{friedlander2012hybrid}, Lemma~2.1].
\end{proof}
\noindent In addition, we obtain
\begin{equation}
    \label{ErrorAMGM}
    \mathbb{E}[\|\mathbf{e}_k^{(t)}\|^{2}] \!\overset{\text{(a)}}{=}\! \mathbb{E}[\|\mathbf{e}_{\mathrm{ds},k}^{(t)} + \mathbf{e}_{\mathrm{com},k}^{(t)}\|^{2}] \!\overset{\text{(b)}}{\leq}\!
    2\left(\mathbb{E} [ \|\mathbf{e}_{\mathrm{ds},k}^{(t)} \|^{2}]+\mathbb{E}[\|\mathbf{e}_{\mathrm{com},k}^{(t)}\|^{2}]\right), 
\end{equation}    
where step (a) is from the expression of $\mathbf{e}_k^{(t)}$ in \eqref{ErrorTwoPart}, and step (b) is from the inequality of arithmetic and geometric means. By plugging \eqref{ErrorAMGM} into \eqref{LemmaEq:upper bound}, we obtain \eqref{F_Upperbound}. 

\section{Proof of Corollary~\ref{Corollary:MSE}}
\label{Appendix:Prf_Corollary:MSE}
From \eqref{ErrorCommunication}, we obtain
\begin{subequations}
    \begin{align}
        & \mathbb{E}\left[\|\mathbf{e}_{\mathrm{com},k}^{(t)}\|^{2}\right]
        \overset{\text{(a)}}{=}  \frac{C}{\left(\sum_{i \in \mathcal{M}_k} Q_{\langle k, i \rangle}\right)^2}\!\left( \sum\nolimits_{i,j\in \mathcal{M}_k} 2\rho_{\langle k, i \rangle, \langle k, j \rangle}^{(t)} Q_{\langle k, i \rangle}Q_{\langle k, j \rangle}  v_k^{(t)} \right.  \notag\\
        &\quad \left. - \zeta_k \sqrt{v_k^{(t)}} \sum\nolimits_{i,j\in \mathcal{M}_k} 2\rho_{\langle k, i \rangle, \langle k, j \rangle}^{(t)} \left( Q_{\langle k, i \rangle} (\mathbf{f}_k^\mathrm{H} \mathbf{H}_{\langle k, j \rangle}^{(t)} \mathbf{u}_{\langle k, j \rangle})^\mathrm{H}   +  Q_{\langle k, j \rangle} \mathbf{f}_k^\mathrm{H} \mathbf{H}_{\langle k, i \rangle}^{(t)} \mathbf{u}_{\langle k, i \rangle} \right) \right.  \notag\\
        &\quad \left. + \zeta_k^2 \left (\sum\nolimits_{l =1}^K \sum\nolimits_{i,j \in \mathcal{M}_{l}}\!2\rho_{\langle l, i \rangle, \langle l, j \rangle}^{(t)}  (\mathbf{f}_k^\mathrm{H} \mathbf{H}_{\langle l, i \rangle}^{(t)} \mathbf{u}_{l,i} )^\mathrm{H}  \mathbf{f}_k^\mathrm{H} \mathbf{H}_{\langle l, j \rangle}^{(t)} \mathbf{u}_{\langle l, j \rangle} + \sigma^{2} \| \mathbf{f}_k \|^2 \right)  \right)  \label{MSEExpand}\\
        & \overset{\text{(b)}}{\geq} 
        \frac{2C v_k^{(t)}}{\left(\sum_{i \in \mathcal{M}_k} Q_{\langle k, i \rangle}\right)^2} 
        \left( \sum\nolimits_{i,j\in \mathcal{M}_k} \rho_{\langle k, i \rangle, \langle k, j \rangle}^{(t)}  Q_{\langle k, i \rangle}Q_{\langle k, j \rangle} \right. \notag\\
        &\quad \left. - \frac{ \left( \sum_{i,j\in \mathcal{M}_k} \rho_{\langle k, i \rangle, \langle k, j \rangle}^{(t)} \left(Q_{\langle k, i \rangle}  (\mathbf{f}_k^\mathrm{H} \mathbf{H}_{\langle k, j \rangle}^{(t)} \mathbf{u}_{\langle k, j \rangle})^\mathrm{H}  + Q_{\langle k, j \rangle} \mathbf{f}_k^\mathrm{H} \mathbf{H}_{\langle k, i \rangle}^{(t)} \mathbf{u}_{\langle k, i \rangle}\right) \right) ^2 }{4 \left(\sum_{l=1}^K \sum_{i,j\in \mathcal{M}_{l}} \rho_{\langle l, i \rangle, \langle l, j \rangle}^{(t)} (\mathbf{f}_k^\mathrm{H} \mathbf{H}_{\langle l, i \rangle}^{(t)} \mathbf{u}_{\langle l, i \rangle})^\mathrm{H} \mathbf{f}_k^\mathrm{H} \mathbf{H}_{\langle l, j \rangle}^{(t)} \mathbf{u}_{\langle l, j \rangle} + \sigma^2  \| \mathbf{f}_k \|^2 /2\right) } 
        \!\right) \label{MSEDerived},
    \end{align}
\end{subequations}
where step (a) is from $\mathbb{E}[r_{\langle k, i \rangle}[c] r_{\langle k, j \rangle}[c]^\dagger] = 2 \rho_{\langle k, i \rangle, \langle k, j \rangle}$ based on Assumption~\ref{asp:g_task}, and step (b) is because $\mathbb{E}[\|\mathbf{e}_{\mathrm{com},k}^{(t)}\|^{2}]$ is a convex quadratic function w.r.t. $\zeta_k$ and the minimizer $\zeta_k^*$ is given by \eqref{opt_c}. 

\section{Proof of Theorem \ref{Theorem:Convergence}}
\label{Appendix:Prf_Theorem:Convergence}
We first derive an upper bound w.r.t. the communication MSE in \eqref{MSEDerived}.
For device $\langle k, i \rangle, \forall k, i$, we have
\begin{align}
    \label{Dkv}
    2 C v_k^{(t)} 
    & = 2 C v_{\langle k, i \rangle}^{(t)} 
     = \sum\nolimits_{d=1}^{D} \left( \mathbb{E}[|g_{\langle k, i \rangle}^{(t)}[d]|^2] -
    |\mathbb{E}[g_{\langle k, i \rangle}^{(t)}[d] ] |^2 \right) \overset{\text{(a)}}{\leq} \mathbb{E}[\|\mathbf{g}_{\langle k, i \rangle}^{(t)}\|^2 ] \notag \\ 
    &\overset{\text{(b)}}{=}\mathbb{E}\bigg[\| \frac{1}{Q_{\langle k, i \rangle}}\sum\nolimits_{n = 1}^{Q_{\langle k, i \rangle}} \nabla f_k (\mathbf{w}^{(t)}_k; \boldsymbol{\xi}_{\langle k, i \rangle,n} ) \|^{2}\bigg] 
    \overset{\text{(c)}}{\leq}\mathbb{E}\bigg[ (\frac{1}{Q_{\langle k, i \rangle}} \sum\nolimits_{n = 1}^{Q_{\langle k, i \rangle}} \|\nabla f_k (\mathbf{w}^{(t)}_k ; \boldsymbol{\xi}_{\langle k, i \rangle,n} ) \|)^2\bigg] \notag \\
    &\overset{\text{(d)}}{\leq} \mathbb{E}\bigg[ \Big(\frac{1}{Q_{\langle k, i \rangle}}  \sum\nolimits_{n = 1}^{Q_{\langle k, i \rangle}} \sqrt{\beta_{1}+\beta_{2}\|\nabla F_k (\mathbf{w}^{(t)}_k )\|^{2}}\Big)^{2}\bigg] =\beta_{1}+\beta_{2}\mathbb{E}[\|\nabla F_k (\mathbf{w}^{(t)}_k )\|^{2}] 
\end{align}
where step (a) is from $\big |\mathbb{E}[g_{\langle k, i \rangle}^{(t)}[d]] \big|^2 \geq 0$; (b) is from the fact that $\mathbb{E}[\mathbf{g}_{\langle k,i \rangle}^{(t)}] = \mathbb{E}[\nabla F_{\langle k,i \rangle} (\mathbf{w}_k^{(t)})]$ and with definition of $ F_{\langle k,i \rangle} (\mathbf{w}_k^{(t)}) $ given below \eqref{LossFuncTaskk}; (c) is from the triangle inequality; and (d) is from \eqref{eq:f_bound}. 
Note that \eqref{MSEDerived} is obtained by the expression of $\mathbb{E}[\|\mathbf{e}_{\mathrm{com},k}^{(t)}\|^{2}]$ in \eqref{ErrorCommunication} as $\zeta_k = \zeta_k^*$. By plugging \eqref{Dkv} into \eqref{MSEDerived}, we obtain an upper bound w.r.t. $\mathbb{E}[\|\mathbf{e}_{\mathrm{com},k}^{(t)}\|^{2}]$:
\begin{align}
    & \mathbb{E}[\|\mathbf{e}_{\mathrm{com},k}^{(t)}\|^{2}]
    \leq  \frac{\beta_1 + \beta_2 \mathbb{E}[ \|\nabla F_k(\mathbf{w}_k^{(t)}) \|^{2} ] }{\left(\sum_{i \in \mathcal{M}_k} Q_{\langle k, i \rangle}\right)^2} \left( \sum_{i,j\in \mathcal{M}_k} \rho_{\langle k, i \rangle, \langle k, j \rangle}^{(t)}  Q_{\langle k, i \rangle}Q_{\langle k, j \rangle} \right. \notag \\
    & \left. -\frac{ \left( \sum_{i,j \in \mathcal{M}_k } \rho_{\langle k, i \rangle, \langle k, j \rangle}^{(t)} \left(Q_{\langle k, i \rangle}  (\mathbf{f}_k^\mathrm{H} \mathbf{H}_{\langle k, j \rangle}^{(t)} \mathbf{u}_{\langle k, j \rangle})^\mathrm{H} + Q_{\langle k, j \rangle} \mathbf{f}_k^\mathrm{H} \mathbf{H}_{\langle k, i \rangle}^{(t)} \mathbf{u}_{\langle k, i \rangle}\right) \right) ^2 }{4 \left(\sum_{l=1}^K \sum_{i,j\in \mathcal{M}_l} \rho_{\langle l, i \rangle, \langle l, j \rangle}^{(t)} (\mathbf{f}_k^\mathrm{H} \mathbf{H}_{\langle l, i \rangle}^{(t)} \mathbf{u}_{\langle l, i \rangle})^\mathrm{H}  \mathbf{f}_k^\mathrm{H} \mathbf{H}_{\langle l, j \rangle}^{(t)}\mathbf{u}_{\langle l, j \rangle} + \sigma^2  \| \mathbf{f}_k \|^2 /2\right) } \right),
    \label{MSEUpperbound}
\end{align}

Next, we derive the upper bound of $\mathbb{E}\left[\mathcal{F}\left(\mathbf{w}^{(t+1)}\right)- \mathcal{F}\left(\mathbf{w}^*\right)\right]$. We take summation on both sides of \eqref{F_Upperbound} to obtain the following inequality:
\begin{equation}
    \label{F_all_Upperbound}
        \mathbb{E}[\mathcal{F}(\!\mathbf{w}^{(t+1)}\!)]
        \leq \mathbb{E}[\mathcal{F}(\!\mathbf{w}^{(t)}\!)] 
        \!-\!\frac{1}{2 \omega} \sum\nolimits_{k = 1}^K \left(\!\mathbb{E} [ \|\nabla F_k(\mathbf{w}_k^{(t)})\|^{2} ] 
        \!-\!2(\mathbb{E}[\|\mathbf{e}_{\mathrm{ds},k}^{(t)}\|^{2}] + \mathbb{E}[\|\mathbf{e}_{\mathrm{com},k}^{(t)}\|^{2}])\!\right).
\end{equation}
where $\omega\!=\!\max_k\!\omega_k$, and $\mu\!=\!\min_k\!\mu_k$. Substituting $\mathbb{E}[\|\mathbf{e}_{\mathrm{ds},k}^{(t)}\|^{2}]$ and $\mathbb{E}[\|\mathbf{e}_{\mathrm{com},k}^{(t)}\|^{2} ]$ respectively with \eqref{ErrorDevicesSelection} and \eqref{MSEUpperbound}, we have an upper bound of the overall loss function $\mathcal{F}_k(\cdot)$ at the round $(\!t\!+\!1\!)$ as
\begin{align}
    \label{Ft+1Upperbound}
        \mathbb{E}[\!\mathcal{F}(\!\mathbf{w}^{(\!t+1\!)}\!)\!]\!
        &\leq\!\mathbb{E}[\!\mathcal{F}(\!\mathbf{w}^{(\!t\!)}\!)\!]\!-\!\frac{1}{2 \omega}\sum_{k = 1}^K\!\left(\!\mathbb{E}[\|\nabla F_k(\!\mathbf{w}_k^{(\!t\!)}\!)\|^{2}]\!(\!1\!-\!2\beta_2
        d_k^{(\!t\!)}\!(\!\mathcal{M}_k,\!\mathbf{f}_k,\!\mathbf{u}_k\!)\!)
        \!-\!2\beta_1\!d_k^{(\!t\!)}\!(\!\mathcal{M}_k,\!\mathbf{f}_k,\!\mathbf{u}_k\!)\!\right)\!, 
\end{align}
where $d_k^{(t)}(\mathcal{M}_k, \mathbf{f}_k, \mathbf{u}_k)$ is defined by \eqref{d}.
Furthermore, based on Assumption~\ref{asp:f_convex}, an upper bound of $\mathbb{E} [\|\nabla F_k(\mathbf{w}_k^{(t)}) \|^{2} ]$ is obtained from [\citenum{friedlander2012hybrid}, eq.~(2.4)] as
\begin{equation}
    \label{NablaFLowerbound}
    \mathbb{E} [ \|\nabla F_k(\mathbf{w}_k^{(t)}) \|^{2} ] \geq 2\mu \mathbb{E} \left[ F_k(\mathbf{w}_k^{(t)}) - F_k\left(\mathbf{w}_k^*\right) \right].
\end{equation}
Plugging \eqref{NablaFLowerbound} into \eqref{Ft+1Upperbound} and subtracting $\mathcal{F}(\mathbf{w}^*)$ on the both sides of \eqref{Ft+1Upperbound}, we obtain \eqref{MSERecursionUpperbound}.

\bibliographystyle{IEEEtran} 
\bibliography{references}

\end{document}